\documentclass[final,1p,authoryear,11pt]{elsarticle}
\usepackage[hidelinks=true]{hyperref}
\usepackage{amssymb}
\usepackage{eurosym}
\usepackage{amsfonts}
\usepackage{amsmath}
\usepackage{amsthm}
\usepackage{verbatim}
\usepackage[hidelinks=true]{hyperref}
\usepackage{color}
\usepackage{graphicx}
\usepackage{caption}
\usepackage{subcaption}
\usepackage{natbib}

\journal{ }
\newtheorem{theorem}{Theorem}[section]

\newtheorem{definition}[theorem]{Definition}

\newtheorem{lemma}[theorem]{Lemma}

\newtheorem{proposition}[theorem]{Proposition}
\newtheorem{remark}[theorem]{Remark}

\numberwithin{equation}{section}
\theoremstyle{plain}

\input{tcilatex}
\begin{document}

\begin{frontmatter}

\title{Optimal Trade Execution Under Endogenous Pressure to Liquidate: Theory and Numerical Solutions}

\author[uniba]{Pavol Brunovsk\'{y}}
\ead{brunovsky@fmph.uniba.sk}

\author[cass]{Ale\v{s} \v{C}ern\'{y}\corref{ca}}
\cortext[ca]{Corresponding author}
\ead{ales.cerny.1@city.ac.uk}

\author[uniba]{J\'{a}n Komadel}
\ead{komadel4@uniba.sk}

\address[uniba]{Department of Applied Mathematics and Statistics, Comenius University Bratislava, 84248 Bratislava, Slovakia}
\address[cass]{Cass Business School, City, University of London, 106 Bunhill Row, London EC1Y 8TZ, UK}

\begin{abstract}
We study optimal liquidation of a trading position (so-called block order or meta-order) in a market with a linear temporary price impact \citep{kyle.85}. 
We endogenize the pressure to liquidate by introducing a downward drift in the unaffected asset price while simultaneously 
ruling out short sales. In this setting the liquidation time horizon becomes a stopping time determined endogenously, 
as part of the optimal strategy. We find that the \emph{optimal} liquidation strategy is consistent with the 
\emph{square-root law} which states that the average price impact per share is proportional to the 
square root of the size of the meta-order \citep{bershova.rakhlin.13,farmer.al.13,donier.al.15,toth.al.16}.

Mathematically, the Hamilton-Jacobi-Bellman equation of our optimization leads to a severely 
singular and numerically unstable ordinary differential equation initial value problem. 
We provide careful analysis of related singular mixed boundary value problems and devise a numerically 
stable computation strategy by re-introducing time dimension into an otherwise time-homogeneous task.
\end{abstract}


\end{frontmatter}

\section{Introduction}

We study optimal liquidation of an infinitely divisible asset when the
execution price is subject to adverse price impact in proportion to the
amount of the asset sold per unit of time, in line with \cite{kyle.85}. The
optimal liquidation strategy trades off expediency against the adverse price
impact caused by a precipitous sale. However, our focus on liquidation is
not fundamental; \emph{mutatis mutandis }one can replace optimal liquidation
with optimal acquisition in what follows. The novelty in our approach is
that we rule out short sales in a falling market. This seemingly small
change has a profound impact on the economics and mathematics of the
problem. How and why this happens is the subject of the ensuing analysis.

Modelling of optimal execution with market impact is relatively new in the
literature, going back to \cite{almgren.chriss.00}, \cite{bertsimas.lo.98}
and \cite{subramanian.jarrow.01}. Classical models \cite{almgren.chriss.00}, 
\cite{bertsimas.lo.98}, envisage a world where the unaffected price of the
asset is a martingale and hence there is no pressure to trade quickly for an
agent with linear utility. In these circumstances the incentive to trade is
given by fiat -- it is assumed that there is a fixed time limit by which the
entire position must be liquidated.

The literature finds that optimal liquidation gives rise to `implementation
shortfall' 
\citep{perold.88}
defined as the gap between the initial market value of the inventory and the
expected revenue of the liquidation strategy; the latter always being lower
due to the price impact. The shortfall itself is formed of two components,
one due to `permanent price impact' and another caused by `temporary
impact'. The former cannot be influenced by the trading strategy, while the
latter determines the optimal strategy and can be made arbitrarily small by
making the liquidation time horizon longer. In this sense having more time
is unambiguously beneficial to the trader.

A second strand of literature, \cite{brown.al.10}, \cite{chen.al.14}, \cite%
{chen.al.15}, identifies the motive to liquidate with a change in market
conditions whereby tighter margin requirements lead to lower permitted
amount of leverage. The change in market conditions occurs at discrete time
points, while the optimal liquidation (deleveraging) is implemented
continuously in time. Here for reasons of tractability the unaffected price
is assumed constant during liquidation, although one could in principle use
the results from the first strand of literature to make the modelling of the
deleveraging phase more realistic.

In this paper we focus on the liquidation phase. Specifically, we study a
situation where the unaffected price may be falling on average, which is
highly plausible in a market with contracting liquidity. One expects that
with the asset price decreasing the implementation shortfall should be more
severe than in the martingale case. Surprisingly, the current literature
finds that far from exhibiting a shortfall the optimal liquidation strategy
may in this case record an expected surplus, see \cite{schied.13}. On closer
inspection one observes that the surplus arises due to short sale of the
asset with subsequent acquisition at deflated price near the end of the
allotted time horizon.

While strategic short sales in a bear market are not entirely implausible we
feel it is important to examine a situation where such short sales are ruled
out. The simplest way to achieve this is to stop the trading once the entire
position has been liquidated. In doing so we recover the classical outcome
from the martingale case whereby the price impact invariably leads to
implementation shortfall. However, in a falling market without short sales
it is no longer true that the shortfall can be made arbitrarily small by
extending the liquidation time horizon.

Introduction of a stopping time is a novel feature in the optimal
liquidation literature with a perfectly divisible asset. Previously, optimal
stopping has appeared only in the context of optimal liquidation of an
indivisible asset, see \cite{mamer.86} and \cite{henderson.hobson.13}.
Although stopping on liquidation automatically precludes short sales, it
does leave open the possibility of further intermediate acquisition. Ex-post
it turns out that intermediate acquisition is not optimal, see Proposition~%
\ref{prop: admissibility} and Theorem~\ref{thm: optimality}. We show that
the presence of the stopping time dramatically changes mathematical
properties of the Hamilton-Jacobi-Bellman equation and leads to a severely
singular and numerically unstable initial value problem. Part of our
research contribution is in providing a comprehensive theoretical and
numerical analysis of this HJB equation and related singular boundary value
problems.

The paper is organized as follows. In Section 2 we survey related literature
and present our model. In Section 3 we discuss reduction of our HJB partial
differential equation (PDE) to an ordinary differential equation (ODE).
Section 4 offers a probabilistic and control-theoretic interpretation of
this reduction. Section 5 describes the singularity of the initial value
problem (IVP) for the ODE of Section 3, while Section 6 shows how to obtain
uniqueness from a related boundary value problem (BVP). In Section 7 we
characterize the optimal strategy and its value function by means of the BVP
of Section 6. In Section 8 we introduce and theoretically analyze a related
PDE BVP which leads to a stable numerical scheme and present numerical
results. Section 9 concludes.

\section{Our model and related literature}

We take the point of view of a trader with inventory $Z$ whose initial value 
$Z(0)>0$ is given. The modelling is based on the premise that there is some
price process $S$ -- often called the `unaffected price'\ -- with
exogenously given dynamics that governs the evolution of the asset price in
the absence of our trading. In our case the unaffected price $S$ is a
geometric Brownian motion%
\begin{equation}
dS(t)=\lambda S(t)dt+\sigma S(t)dB(t),  \label{0.4}
\end{equation}%
where $B$ is a Brownian motion in its natural filtration.

The inventory attracts interest rate $r,$ which becomes a storage cost when $%
r<0$. We assume that the inventory is sold off continuously at a
(stochastic) rate $v:=-dZ/dt$ so that $v$ represents the amount of inventory
sold per unit of time. Consequently, the inventory dynamics read 
\begin{equation}
dZ(t)=\left( rZ(t)-v(t)\right) dt.  \label{0.2c}
\end{equation}

Let $T(Z=0)$ be the first time when the entire inventory is disposed of. For
a given pair of initial values 
\begin{equation}
s=S(0),z=Z(0),  \label{eq: s,z}
\end{equation}%
the expected discounted revenue from the disposal of the asset is given by%
\begin{equation}
J(s,z,v)=E_{s,z}\left[ \int_{0}^{T(Z=0)}e^{-\rho t}(S(t)-\eta v(t))v(t)dt%
\right] ,  \label{0.2b}
\end{equation}%
where $S-\eta v$ is the `affected price'\ of the asset. In our setting $\eta 
$ measures the strength of `temporary impact'\ the selling speed $v$ has on
the price. The discount factor $\rho $ captures the opportunity cost of not
holding alternative assets. The entire model is based on \cite{cerny.99}.

The task is to find optimal liquidation strategy $v$ that maximizes 
\begin{equation}
V(s,z):=\sup_{v\in \mathcal{A}}J(s,z,v).  \label{eq: (L)}
\end{equation}%
We say that $v$ is an admissible control, and write $v\in \mathcal{A}$, if
process $v$ is predictable,%
\begin{equation}
E\left[ \int_{0}^{t}\left\vert v(s)\right\vert ^{m}ds\right] <\infty \text{
for all }t>0\text{ and }m=1,2,\ldots ,  \label{eq: adm1}
\end{equation}%
and 
\begin{equation}
E\left( \int_{0}^{T(Z=0)}e^{-\rho s}\left\vert v(t)(S(t)-\eta
v(t))\right\vert dt\right) <\infty .  \label{eq: adm2}
\end{equation}

The optimization in our model can be seen, for specific parameter choices,
as a special case of \cite{ankirchner.kruse.13}, \cite{forsyth.al.12} and 
\cite{schied.13}, with the crucial difference that in our case the
liquidation time horizon is endogenous. We make a standing assumption that
the time discounting is stronger than the expected appreciation and the
interest on the asset combined, 
\begin{equation}
\rho >\lambda +r.  \label{eq: rho ineq}
\end{equation}

To conclude this section we wish to make several observations that justify
the choice of our modelling framework. The extant literature contains a
number of variations on the model presented above. The trading may be
discrete, rather than continuous, the unaffected price $S$ may be specified
differently and the optimization criterion may involve a utility function.
In common, existing models assume $T$ is fixed and exogenously given.

The most commonly considered specification for the affected price reads 
\begin{equation}
\tilde{S}:=S-\gamma (Z(0)-Z_{-}-\frac{1}{2}\Delta Z)-\eta _{1}v+\eta
_{2}\Delta Z,\text{ }  \label{eq: impact}
\end{equation}%
where $\gamma (Z(0)-Z_{-}-\frac{1}{2}\Delta Z)$ is the `permanent'\ price
impact\footnote{%
Note that this classification of permanent impact differs subtly from the
one used in \cite{almgren.chriss.00} and subsequent literature. In our
classification permanent impact has no strategic effect on optimal execution
when $S$ is a martingale.} while $\eta _{1}v$ and $\eta _{2}\Delta Z,$
respectively, are known as `temporary'\ price impacts in the continuous-time
and discrete-time literature, respectively. It is assumed either that there
is a finite number of fixed dates $\{t_{i}\}_{i=1}^{N}$ where $Z$ is allowed
to jump (discrete-time models) or that $Z$ changes continuously at a
stochastic time rate $-v$ (continuous-time models). In each case $Z$ is
taken to be a predictable semimartingale with left limit process $Z_{-}$ and
jumps $\Delta Z=Z-Z_{-}$. Models in this category include \cite%
{ankirchner.al.16}, \cite{brown.al.10}, \cite{chen.al.14}, \cite%
{gatheral.schied.11}, \cite{schied.13}, \cite{schied.schoneborn.09} \cite%
{ting.al.07} in continuous time and{\footnotesize \ }\cite{almgren.chriss.00}%
, \cite{bertsimas.lo.98} in discrete time. Other impact specifications can
be found, for example, in \cite{chen.al.15}, \cite{cheridito.sepin.14}, \cite%
{forsyth.11}, \cite{lorenz.almgren.11}, \cite{subramanian.jarrow.01} and 
\cite{ting.al.07}.

The revenue $R(T)$ from liquidation over a fixed time horizon $T$ is given
by $R(T):=\int_{0}^{T}-\tilde{S}(t)dZ(t)$. When the unaffected asset price
process $S$ is a martingale, integration by parts together with suitable
boundedness of $Z$ and boundary condition $Z(T)=0$ yields%
\begin{equation*}
E[R(T)]=Z(0)S(0)-\frac{\gamma }{2}Z(0)^{2}-\eta
_{1}\int_{0}^{T}v^{2}(t)dt-\eta _{2}\sum_{i=1}^{N}\left( \Delta
Z(t_{i}\right) ^{2}.
\end{equation*}%
This equality offers several important insights:

\begin{enumerate}
\item Permanent impact (as defined here) has no strategic influence and in
the absence of temporary impact ($\eta _{1}=\eta _{2}=0$) any strategy $Z$
is optimal. The expected implementation shortfall $Z(0)S(0)-E[R(T)]=\frac{%
\gamma }{2}Z(0)^{2}$ is strictly positive.

\item With temporary impact it is optimal to liquidate at a constant rate,
regardless of the strength of the permanent impact. The additional
implementation shortfall equals $\eta _{1}Z(0)^{2}/T$ in continuous time and 
$\eta _{2}Z(0)^{2}/N$ in discrete time, respectively.
\end{enumerate}

These observations suggest that temporary impact is responsible for the
majority of strategic interaction also in the drifting market and we
conjecture that the optimal \emph{strategy} will therefore not change
dramatically when the permanent impact is included. This is not to say that
the implementation \emph{shortfall} would be unaffected by the presence of
permanent impact. Given the complexity of analysis to follow and the likely
marginal gains to our understanding from the presence of permanent impact on
the optimal trading strategy we feel justified in leaving out the permanent
impact from our analysis.

More recent studies, excellently summarized in \cite{gatheral.10}, consider
an intermediate form of impact where the execution price is given by the
formula%
\begin{equation*}
S_{t}-\int_{0}^{t}f(v_{u})G(t-u)du.
\end{equation*}%
Kernel $G$ is called the \emph{resiliency} of the market and the two extreme
cases, permanent impact and temporary impact, correspond to $G$ being
constant or $G$ being the Dirac delta function, respectively. In \cite%
{gatheral.10} a case is made for a combination of power impact function, $%
f(v)=v^{\delta }$, with power law resiliency $G(x)=x^{-\gamma }$, $\delta
+\gamma \geq 1$, the latter tending to a Dirac delta function as $\gamma
\searrow 0$. We note that our setup corresponds to the limiting case $\delta
=1,\gamma =0$ and we leave the analysis of the general impact function $f$
with general resiliency $G$ in the setup of this paper to future research.

\section{HJB equation and dimension reduction}

The value function $V$ defined in (\ref{eq: (L)})\ formally solves the
Hamilton-Jacobi-Bellman partial differential equation%
\begin{equation*}
\sup_{v}\{\frac{1}{2}s^{2}\sigma ^{2}V_{ss}+\lambda sV_{s}+(rz-v)V_{z}-\rho
V+v(s-\eta v)\}=0,\ s>0,z>0,
\end{equation*}%
with formal optimal control%
\begin{equation*}
v^{\ast }=\frac{s-V_{z}}{2\eta },
\end{equation*}%
giving rise to a quasilinear second order PDE 
\begin{equation}
\frac{1}{2}s^{2}\sigma ^{2}V_{ss}+\lambda sV_{s}+rzV_{z}-\rho V+\frac{%
(s-V_{z})^{2}}{4\eta }=0,  \label{eq: HJB}
\end{equation}%
with an initial condition%
\begin{equation}
V(s,0)=0.  \label{eq: V(s,0)}
\end{equation}

The self-similarity 
\begin{equation}
V(s,z)=s^{2}u(x)/(\eta \sigma ^{2}),\ x=\eta \sigma ^{2}z/s,  \label{scaling}
\end{equation}%
reduces (\ref{eq: HJB}, \ref{eq: V(s,0)}) to an initial value problem (IVP)
for an ordinary differential equation (ODE)%
\begin{align}
x^{2}u^{\prime \prime }& =axu^{\prime }+bu-(u^{\prime }-1)^{2}/2,\ x>0,
\label{0} \\
u(0)& =0,  \label{0.1}
\end{align}%
where 
\begin{equation}
a:=2(\lambda -r+\sigma ^{2})/\sigma ^{2},\ b:=-2(2\lambda -\rho +\sigma
^{2})/\sigma ^{2}.  \label{eq: abc def}
\end{equation}%
The self-similarity reduces a problem with 7 independent parameters $\rho
,\lambda ,r,\sigma ,\eta ,s\equiv S(0)$ and $z\equiv Z(0)$ to a problem with
just three parameters: $a,$ $b$ and $x:=\eta \sigma ^{2}z/s$.

\section{Probabilistic interpretation of self-similarity}

We begin by restating the HJB equation (\ref{eq: HJB}) in its variational
form,%
\begin{equation*}
\sup_{v(0)}\left\{ \text{drift}_{0}\left( e^{-\rho t}V(S,Z)\right)
+v(0)\left( S(0)-\eta v(0)\right) \right\} =0.
\end{equation*}%
Plug in the self-similarity form of the value function $V(S,Z)=S^{2}u\left(
\eta \sigma ^{2}Z/S\right) /(\eta \sigma ^{2})$ and rearrange to obtain%
\begin{equation*}
\sup_{v(0)}\left\{ \text{drift}_{0}\left( e^{-\rho t}\frac{S^{2}}{S(0)^{2}}%
u\left( \frac{\eta \sigma ^{2}Z}{S}\right) \right) +\eta \sigma ^{2}\frac{%
v(0)}{S(0)}\left( 1-\eta \frac{v(0)}{S(0)}\right) \right\} =0.
\end{equation*}%
The next steps involve i) changing measure to $\hat{P}$ given by $\frac{d%
\hat{P}_{t}}{dP_{t}}=\frac{S(t)^{2}}{S(0)^{2}}e^{-\left( 2\lambda +\sigma
^{2}\right) t}$ where $\hat{P}_{t}$ and $P_{t}$ are restrictions of $\hat{P}$
and $P$ to $\mathcal{F}_{t}$; ii) defining a new state variable $X:=\eta
\sigma ^{2}Z/S$; and iii) reparametrizing the control to $g:=\eta v/S,$
which yields%
\begin{equation}
\sup_{g(0)}\left\{ \widehat{\text{drift}}_{0}\left( e^{(2\lambda +\sigma
^{2}-\rho )t}u\left( X\right) \right) +\sigma ^{2}g(0)\left( 1-g(0)\right)
\right\} =0.  \label{eq: HJB2}
\end{equation}

The It\={o} formula for $X$ reads%
\begin{equation*}
dX=\left( rX-\sigma ^{2}g\right) dt+X\left( -\frac{dS}{S}+\frac{d[S,S]}{S^{2}%
}\right) ,
\end{equation*}%
while from the Girsanov theorem we obtain $\widehat{\text{drift}}\left( 
\mathcal{L}(S)\right) =\lambda +2\sigma ^{2}$, which implies 
\begin{equation}
dX=\left( (r-\lambda -\sigma ^{2})X-\sigma ^{2}g\right) dt+\sigma ^{2}Xd\hat{%
B},  \label{eq: SDE X}
\end{equation}%
where $\hat{B}:=-\mathcal{L}(S)+(\lambda +2\sigma ^{2})t$ is a Brownian
motion under $\hat{P}$ and $\mathcal{L}(S)$ denotes the stochastic logarithm
of $S,$ $d\mathcal{L}(S)=dS/S$. In the final step we perform a time change
from $t$ to $\sigma ^{2}t,$ defining $\hat{X}(t):=X(t/\sigma ^{2})$ and $%
\hat{W}(t):=\sigma ^{2}\hat{B}(t/\sigma ^{2})$. This yields the dynamics%
\begin{equation}
d\hat{X}=\left( \frac{r-\lambda -\sigma ^{2}}{\sigma ^{2}}\hat{X}-g\right)
dt+\hat{X}d\hat{W},  \label{eq: SDE Xhat}
\end{equation}%
while (\ref{eq: HJB2}) changes to 
\begin{equation}
\sup_{g(0)}\left\{ \widehat{\text{drift}}_{0}\left( \exp \left( \frac{%
2\lambda +\sigma ^{2}-\rho }{\sigma ^{2}}t\right) u(\hat{X})\right)
+g(0)\left( 1-g(0)\right) \right\} =0.  \label{eq: HJB3}
\end{equation}

With (\ref{eq: SDE Xhat}) in hand the optimality condition (\ref{eq: HJB3})
explicitly reads%
\begin{eqnarray*}
0 &=&\frac{1}{2}x^{2}u^{\prime \prime }(x)+\frac{r-\lambda -\sigma ^{2}}{%
\sigma ^{2}}xu^{\prime }(x) \\
&&+\frac{2\lambda +\sigma ^{2}-\rho }{\sigma ^{2}}u(x)+\frac{1}{4}\left(
1-u^{\prime }(x)\right) ^{2},
\end{eqnarray*}%
and the (formal) optimal control equals $g=(1-u^{\prime }(\hat{X}))/2$. It
is furthermore clear that (\ref{eq: HJB3}) itself is a HJB\ equation of an
optimal control problem 
\begin{equation}
u(x)=\sup_{g}\widehat{E}_{x=\hat{X}(0)}\left[ \int_{0}^{T(\hat{X}=0)}\exp
\left( -\frac{\rho -2\lambda -\sigma ^{2}}{\sigma ^{2}}t\right)
g(t)(1-g(t))dt\right] ,  \label{eq: u(x) optimization}
\end{equation}%
with $\hat{P}$-dynamics of $\hat{X}$ given by (\ref{eq: SDE Xhat}).

Note that the time in the transformed problem (\ref{eq: u(x) optimization})
is measured in terms of cumulative variance of the log return of the
unaffected price, that is in `variance years'. One variance year\
corresponds to the physical time $t$ it takes to make $\sigma ^{2}t=1$. With 
$\sigma =0.2$ one variance year is therefore equal to 25 calendar years. The
new state variable $\hat{X}=\eta \sigma ^{2}Z/S$ corresponds to the size of
temporary price impact as a percentage of current price, assuming inventory $%
Z$ is completely liquidated at a constant rate over one variance year.

\section{Singular initial value problem IVP$_{0}$}

Hereafter we refer to the IVP (\ref{0}, \ref{0.1}) as $\mathrm{IVP}_{0}$.
Note $a+b>0$ if and only if our standing assumption (\ref{eq: rho ineq}) $%
\rho >r+\lambda $ holds. It has been shown in \cite{brunovsky.al.13} that $%
\mathrm{IVP}_{0}$ is highly degenerate at $0.$ For $a+b>0$ $\mathrm{IVP}_{0}$
has infinitely many solutions with identical asymptotics near $0$ given by
the formal power series%
\begin{equation}
h_{n}(x)=x-\frac{2}{3}\sqrt{2(a+b)}x^{3/2}+\sum_{i=2}^{n}k_{i}x^{1+i/2},%
\text{ }n\in \mathbb{N},  \label{eq: u(x) series}
\end{equation}%
where $k_{i}$ are obtained recursively from 
\begin{eqnarray*}
k_{n+1} &=&\frac{1}{3\left( n+3\right) k_{1}}\bigg[k_{n}\left( \left(
n+2\right) (2a-n)+4b\right) \\
&&-\frac{1}{2}\sum_{j=1}^{n-1}(3+j)(n-j+3)k_{j+1}k_{n-j+1}\bigg].
\end{eqnarray*}%
The series itself has zero\textbf{\ }radius of convergence for 
\begin{equation*}
6a+4b-3=:K_{1}>0>K_{2}:=6a+2b-9,
\end{equation*}%
see {\cite[Remark 2]{quittner.15}}. Asymptotic expansion of derivatives of $%
u(x)$ is obtained by formal differentiation of the series in (\ref{eq: u(x)
series}), \emph{ibid} Theorem 1. Whenever $K_{1}=0$ or $K_{2}=0$ the power
series ends at the 3rd element and constitutes a genuine solution of $%
\mathrm{IVP}_{0}$. This solution, however, is just one from a continuum and
does not represent the optimal value function.

The highly degenerate nature of $\mathrm{IVP}_{0}$ does not stem from the
singularity of the linear terms in the ODE, which is well known and rather
innocuous in the context of the Black-Scholes model, but from the
singularity of the non-linear term. \cite{liang.09} studies singular IVPs of
the form $u^{\prime \prime }=x^{-1}f(x,u,u^{\prime })$ where $f$ is
continuous. Note that the linear part of our ODE, $ax^{-1}u^{\prime },$
belongs to Liang's category, but the non-linear term $x^{-2}\left( u^{\prime
}-1\right) ^{2}/2$ does not.

Liang, too, observes multiplicity of solutions, but this multiplicity is
less pronounced than in our case. In Liang's work $u(0)$ and $u^{\prime }(0)$
uniquely determine the first $\left\lceil \gamma \right\rceil $ derivatives
of the solution, where $\gamma :=\frac{\partial }{\partial u^{\prime }}%
f(0,u(0),u^{\prime }(0))>0$, and, for non-integer $\gamma $, the solution
becomes unique once the coefficient by $x^{\gamma }$ has been specified.
Therefore, in Liang's case all solutions differ asymptotically by a multiple
of $x^{\gamma }$ near $0$.

In contrast, $\mathrm{IVP}_{0}$ has a continuum of solutions that differ
asymptotically by 
\begin{equation*}
x^{\alpha }\exp \left( -\beta /\sqrt{x}\right) ,
\end{equation*}%
where 
\begin{equation*}
\alpha :=2-\frac{2}{3}b,\qquad \beta :=\sqrt{8(a+b)},
\end{equation*}%
see {\cite[Theorems 2-5]{quittner.15}. }These solutions invariably share
their power series asymptotics to an arbitrary order as $x\searrow 0$. A
uniqueness result relevant for the current paper can be summarized as
follows:

\begin{proposition}
\label{bcw} Under the assumption $a+b>0$ there is a unique solution of $%
\mathrm{IVP}_{0}$ denoted by $u_{\infty }$ satisfying $u_{\infty }\in 
\mathcal{C}^{0}[0,\infty )\times \mathcal{C}^{2}(0,\infty )$,%
\begin{equation}
0\leq u_{\infty }(x)\leq x\text{ for }x>0.  \label{uhoredolu}
\end{equation}%
The solution $u_{\infty }$ further satisfies $u_{\infty }^{\prime
}(0)=1,u_{\infty }^{\prime }(x)>0,u_{\infty }^{\prime \prime
}(x)<0,u_{\infty }^{\prime \prime \prime }(x)>0$ for all $x>0$ as well as $%
u_{\infty }^{\prime }(x)\searrow 0$ for $x\rightarrow \infty $.
\end{proposition}

\begin{proof}
See Proposition 5.1 in \cite{brunovsky.al.13}.
\end{proof}

Proposition \ref{prop: alternatives} reveals certain qualitative
characteristics of solutions of $\mathrm{IVP}_{0}$ which can be observed
empirically whenever an unstable numerical scheme is employed.

\begin{proposition}
\label{prop: alternatives}Any solution of (\ref{0}) on $(\alpha ,\beta )$
with $0\leq \alpha <\beta \leq \infty $ falls into one and only one of the
following categories:

i) $u$ is constant;

ii) $u$ is strictly concave on $(\alpha ,\beta )$;

iii) $u$ is strictly convex on $(\alpha ,\beta )$;

iv)\ there is $x_{0}\in (\alpha ,\beta )$ such that $u$ is strictly concave
on $(\alpha ,x_{0})$, strictly convex on $(x_{0},\beta )$ and $u^{\prime
}(x)\geq u^{\prime }(x_{0})>0$ for all $x\in \left( \alpha ,\beta \right) $;

v)\ there is $x_{0}\in (\alpha ,\beta )$ such that $u$ is strictly convex on 
$(\alpha ,x_{0})$, strictly concave on $(x_{0},\beta )$ and $u^{\prime
}(x)\leq u^{\prime }(x_{0})<0$ for all $x\in \left( \alpha ,\beta \right) $.
\end{proposition}

\begin{proof}
The conclusions follow readily from \cite{brunovsky.al.13}, Lemma 4.1,
applied to the equation 
\begin{equation}
x^{2}y^{\prime \prime }=(1+(a-2)x-y))y^{\prime }+\left( a+b\right) y,
\label{eq: dODE}
\end{equation}%
with $y=u^{\prime }$, obtained by differentiation and re-arrangement of (\ref%
{0}).
\end{proof}

\section{Boundary value problem BVP$_{[0,\infty )}$}

In the context of the present paper it turns out to be advantageous to view
Proposition \ref{bcw} as a solution to a certain \emph{boundary value problem%
} (BVP). We write $u^{\prime }(\infty ):=\lim_{x\rightarrow \infty
}u^{\prime }(x)$ whenever the limit on the right-hand side exists and
complement the Dirichlet-type boundary condition $u(0)=0$ with a
Neumann-type boundary condition%
\begin{equation}
u^{\prime }\left( \infty \right) =0.  \label{eq: u'(infty)=0}
\end{equation}%
Hereafter we refer to the mixed boundary value problem (\ref{0}, \ref{0.1}, %
\ref{eq: u'(infty)=0}) as $\mathrm{BVP}_{[0,\infty )}$. It is seen below
that the right-hand boundary condition (\ref{eq: u'(infty)=0}) uniquely
determines the solution found in Proposition \ref{bcw}.

\begin{proposition}
\label{prop: bvp_infty}Under the assumption $a+b>0$ $\mathrm{BVP}_{[0,\infty
)}$ has a unique solution which additionally satisfies $u^{\prime }(0)=1$, $%
u^{\prime }>0$, $u^{\prime \prime }<0$, $u^{\prime \prime \prime }>0$, as
well as $0\leq u(x)\leq x$.
\end{proposition}

\begin{proof}
$\mathrm{BVP}_{[0,\infty )}$ possesses at least one solution, namely the
solution identified in Proposition \ref{bcw}. Below we will prove uniqueness
by showing that any solution of $\mathrm{BVP}_{[0,\infty )}$ must also
satisfy $0\leq u(x)\leq x$. By Lemma 3.1 in \cite{brunovsky.al.13} any local
solution of the $\mathrm{IVP}_{0}$ satisfies $\lim_{x\rightarrow
0_{+}}u^{\prime }(x)=u^{\prime }(0)=1$. Consider now the alternatives in
Proposition \ref{prop: alternatives} with $\alpha =0$ and $\beta =\infty $.
Since any solution of $\mathrm{BVP}_{[0,\infty )}$ also solves $\mathrm{IVP}%
_{0}$ it cannot fall into the constant alternative i). Similarly, it cannot
fall into category iii) with $u^{\prime \prime }>0$ since $u^{\prime }(0)=1$
then implies $u^{\prime }(\infty )\geq 1$. Alternatives iv) and v) also
imply $u^{\prime }(\infty )\neq 0$. Therefore only category ii) remains as a
possible alternative. One thus obtains $u^{\prime \prime }<0$ globally,
therefore $u^{\prime }$ is decreasing and $u^{\prime }(\infty )=0$ implies $%
u^{\prime }\geq 0.$ We have thus proved $0\leq u^{\prime }\leq 1$ and on
integrating one obtains $0\leq u\leq x$. This shows uniqueness by
Proposition \ref{bcw}.
\end{proof}

The paper \cite{brunovsky.al.13} left two questions open. The first is
whether the value function $V$ generated by the solution $u_{\infty }$ of $%
\mathrm{BVP}_{[0,\infty )}$ from Proposition \ref{bcw} via equation (\ref%
{scaling}) is indeed the value function of the optimization problem (\ref%
{eq: (L)}). The second question concerns numerical computation of the
solution to $\mathrm{BVP}_{[0,\infty )}$. We address both questions in turn,
the former in Section 5 and the latter in Section 6.

\section{Optimality}

\label{sect: opt}In this section we establish the precise connection between
the boundary value problem $\mathrm{BVP}_{[0,\infty )}$ and the optimal
control and value function for the liquidation problem (\ref{eq: (L)}). We
begin by formulating a natural sufficient condition for admissibility and
investigate under what circumstances it is admissible to pursue further
acquisition of the asset to be liquidated, $v<0$.

\begin{proposition}
\label{prop: admissibility}Under the assumption (\ref{eq: rho ineq}) any
predictable control $v$ satisfying $S(t)/\eta \geq v(t)\geq 0$ is
admissible. If additionally 
\begin{equation}
\rho >\lambda ^{+}+r^{+},  \label{eq: rho ineq +}
\end{equation}%
where $x^{+}:=\max (x,0),$ then any predictable control $v$ satisfying $%
S(t)/\eta \geq v(t)\geq -K$ for some $K>0$ is also admissible.
\end{proposition}

\begin{proof}
i) We have $\left\vert v(t)\right\vert ^{m}\leq \left( S(t)/\eta \right)
^{m}+K^{m}$ and since $S$ is a GBM\ this implies $E\left[ \int_{0}^{t}\left\vert v(s)\right\vert ^{m}ds\right] <\infty $ for any finite $t$ and any $%
m\in \mathbb{N}$ which proves (\ref{eq: adm1}).

ii) To prove (\ref{eq: adm2}) first note that $v(t)\geq -K$ implies 
\begin{equation}
Z(t)\leq ze^{rt}+K\frac{e^{rt}-1}{r}.  \label{eq: Z estimate}
\end{equation}%
To show integrability of the value function we first obtain an estimate of
the integrand 
\begin{equation*}
\left\vert v(t)(S(t)-\eta v(t))\right\vert \leq \left(
v(t)^{+}+v(t)^{-}\right) \left( S(t)+\eta v(t)^{-}\right) \leq \left(
K+v(t)^{+}\right) \left( S(t)+\eta K\right) ,
\end{equation*}%
which for any bounded stopping time $\tau $ yields%
\begin{multline*}
E\left[ \int_{0}^{\tau }e^{-\rho t}\left\vert v(t)(S(t)-\eta
v(t))\right\vert dt\right] \\
\leq K\int_{0}^{\tau }e^{-\rho t}(E[S(t)]+\eta K)dt+E\left[ \int_{0}^{\tau
}e^{-\rho t}v(t)^{+}(S(t)+\eta K)dt\right] \\
\leq \underbrace{K\int_{0}^{\infty }e^{-\rho t}(se^{\lambda t}+\eta K)dt}%
_{C<\infty }+E\left[ \int_{0}^{\tau }e^{-\rho t}v(t)^{+}(S(t)+\eta K)dt%
\right] .
\end{multline*}%
Continue with the integral inside the expectation in the second term,
letting $W(t)=\int_{0}^{t}v(s)^{+}ds,$ and integrating by parts. In
preparation note $dZ(t)=\left( rZ(t)-v(t)\right) dt$ which together with (%
\ref{eq: Z estimate}) implies for any bounded stopping time $\tau \leq T$ 
\begin{eqnarray}
0 &\leq &W(\tau )=\int_{0}^{\tau }v(t)^{+}dt=\int_{0}^{\tau
}v(t)^{-}dt+\int_{0}^{\tau }rZ(t)dt+z-Z\left( \tau \right)  \notag \\
&\leq &K\tau +\int_{0}^{\tau }r\left( ze^{rt}+K\frac{e^{rt}-1}{r}\right)
dt+z=:g(\tau ).  \label{eq: W estimate}
\end{eqnarray}%
Integration by parts yields%
\begin{eqnarray*}
\int_{0}^{\tau }e^{-\rho t}v(t)^{+}(S(t)+\eta K)dt &=&e^{-\rho \tau }W(\tau
)\left( S(\tau )+\eta K\right) \\
&&+\rho \int_{0}^{\tau }e^{-\rho t}W(t)\left( S(t)+\eta K\right)
dt-\int_{0}^{\tau }e^{-\rho t}W(t)dS(t).
\end{eqnarray*}%
We continue with the second term on the right-hand side. Let $dM(t)=e^{-\rho
t}W(t)S(t)dB(t)$ then $E[[M,M]_{t}]=E\left[ \int_{0}^{t}e^{-2\rho
l}W^{2}(l)S^{2}(l)dl\right] \leq s^{2}\int_{0}^{t}g^{2}(l)e^{2(\lambda
+\sigma ^{2}-\rho )l}dt<\infty $ with $g$ from (\ref{eq: W estimate}) which
implies that $M$ is a (square-integrable) martingale. Hence for any bounded
stopping time $\tau $%
\begin{equation*}
E\left[ \int_{0}^{\tau }e^{-\rho t}W(t)dS(t)\right] =E\left[ \int_{0}^{\tau
}e^{-\rho t}\lambda W(t)S(t)dt\right] .
\end{equation*}%
Pulling everything together%
\begin{equation*}
E\left[ \int_{0}^{\tau }e^{-\rho t}\left\vert v(t)(S(t)-\eta
v(t))\right\vert dt\right] \leq C+E\left[ \int_{0}^{\tau }e^{-\rho
t}v(t)^{+}(S(t)+\eta K)dt\right] .
\end{equation*}%
The right-hand side is bounded for $K=0$ under the standing assumption (\ref%
{eq: rho ineq}). This is also true for $K>0$ if additionally $\rho >0,\rho
>\lambda $ and $\rho >r.$ The last three inequalities together with the
standing assumption (\ref{eq: rho ineq}) are equivalent to (\ref{eq: rho
ineq +}). Letting $\tau $ increase to $T$ we have by monotone convergence 
\begin{equation*}
E\left[ \int_{0}^{T}e^{-\rho t}\left\vert v(t)(S(t)-\eta v(t))\right\vert dt%
\right] <\infty .
\end{equation*}
\end{proof}

The next theorem characterizes the optimal liquidation strategy and the
corresponding value function. The inequality $V(s,z)\leq sz$ confirms the
initial intuition that without short sales the implementation shortfall $%
sz-V(s,z)$ must be positive. We note that due to $0\leq u_{\infty }^{\prime
}(x)\leq 1$ we have $v^{\ast }(t)\geq 0,$ i.e. it is \emph{not} optimal to
buy more of the liquidated asset, even when (for $\rho >\lambda ^{+}+r^{+}$)
strategies that involve further purchases are admissible.

\begin{theorem}
\label{thm: optimality}Assume (\ref{eq: rho ineq}). Let $u_{\infty }$ be the
unique solution of $\mathrm{BVP}_{[0,\infty )}$, with $a,b$ given by (\ref%
{eq: abc def}). Then the function $V(s,z):=\frac{s^{2}}{\eta \sigma ^{2}}%
u_{\infty }\left( \eta \sigma ^{2}\frac{z}{s}\right) \leq sz$ is the value
function of the optimization (\ref{eq: (L)}) and 
\begin{equation}
v^{\ast }(t):=\frac{1}{2\eta }\left( S(t)-V_{z}(S(t),Z^{\ast }(t))\right) =%
\frac{S(t)}{2\eta }\left( 1-u_{\infty }^{\prime }\left( \eta \sigma ^{2}%
\frac{Z^{\ast }(t)}{S(t)}\right) \right) \geq 0  \label{eq: vstar}
\end{equation}%
is the optimal control among all admissible controls $\mathcal{A}$ defined
in equations (\ref{eq: adm1},\ref{eq: adm2}).
\end{theorem}

\begin{proof}
To prove the theorem we apply the `Verification' Theorem IV.5.1 of \cite%
{fleming.soner.06}. To this end, we have to check the following:

\begin{enumerate}
\item[(i)] $V(s,z)$ is $\mathcal{C}^{2}\left( (0,\infty )\times (0,\infty
)\right) \cap \mathcal{C}^{0}\left( [0,\infty )\times \lbrack 0,\infty
)\right) $ and satisfies 
\begin{equation*}
|V(s,z)|\leq K(1+|(s,z)|^{m})
\end{equation*}%
for some $m>0,\ K>0$;

\item[(ii)] $\lim \sup_{t\rightarrow \infty }E_{(s,z)}\left[ I_{t\leq
T\left( Z=0\right) }e^{-\rho t}V(S(t),Z(t))\right] \geq 0$ for all
admissible controls, where $s:=S(0)$ and $z:=Z(0)$;

\item[(iii)] For any finite time $t$%
\begin{equation*}
\lim_{t\rightarrow \infty }e^{-\rho t}E_{(s,z)}\left[ I_{t\leq T\left(
Z^{\ast }=0\right) }V(S(t),Z^{\ast }(t))\right] =0,
\end{equation*}%
$(S(t),Z^{\ast }(t))$ being the solution of 
\begin{align*}
dS(t)& =\lambda S(t)dt+\sigma S(t)dB(t), \\
dZ^{\ast }(t)& =\left( rZ^{\ast }(t)-\frac{S(t)}{2\eta }\left( 1-u^{\prime
}\left( \eta \sigma ^{2}\frac{Z^{\ast }(t)}{S(t)}\right) \right) \right) dt.
\end{align*}
\end{enumerate}

The regularity properties as well as the estimates of (i) and (ii) are
immediate consequences of the properties of $u_{\infty }$, which in
particular imply%
\begin{equation}
0\leq V(s,z)\leq sz.  \label{eq: V(s,z) le sz}
\end{equation}

The estimate (\ref{eq: Z estimate}) gives $Z^{\ast }(t)\leq ze^{rt}$ which
in combination with inequality (\ref{eq: V(s,z) le sz}) and standing
assumption (\ref{eq: rho ineq}) yields%
\begin{eqnarray*}
0 &\leq &e^{-\rho t}E_{(s,z)}\left[ I_{t\leq T\left( Z^{\ast }=0\right)
}V(S(t),Z^{\ast }(t))\right]  \\
&\leq &e^{-\rho t}ze^{rt}E_{(s,z)}\left[ S(t)\right] =sze^{\left( r+\lambda
-\rho \right) t}\searrow 0.
\end{eqnarray*}%
This proves item (iii).
\end{proof}

{Observe that the optimal control deviates from the myopic strategy of
maximizing the integrand of the objective function }$v_{\mathrm{myopic}%
}(t):=S(t)/(2\eta )${. In addition to the instantaneous impact on the
execution price the current liquidation rate also affects future levels of
the inventory }$Z${. In (\ref{eq: vstar}) the optimal strategy at time $t$
differs from }$v_{\mathrm{myopic}}(t)${\ by the amount $-V_{z}(S(t),Z^{\ast
}(t))$, which is the marginal value of the optimal revenue with respect to
the size of the remaining inventory. }It follows that taking proper account
of the role of future inventory level reduces the selling rate. By
Proposition \ref{bcw}, $u_{\infty }^{\prime }$ is positive and decreasing to
zero and so is $V_{z}(s,z)$ in $z$ and therefore for large values of $%
Z^{\ast }(t)$ the selling rate is very close to the myopic strategy. For
small values of $Z^{\ast }(t)$ the optimal rate of trading is non-linear,
roughly proportional to $\sqrt{Z}$ as can be seen from the asymptotic
expansion (\ref{eq: u(x) series}) and the formula for the optimal trading
rate (\ref{eq: vstar}).

We remark that the classical martingale case with $\rho =\lambda =r=0$ and
fixed time horizon $T$ yields constant optimal liquidation speed $v^{\ast
}=Z(0)/T$. The resulting price impact per share, for fixed $T$, is
proportional to $Z(0)$ which is not consistent with broad empirical evidence
that indicates power dependence roughly proportional to $\sqrt{Z(0)}$.

When estimating price impact empirically, an assumption has to be made about
the rate of trading. In \cite{almgren.al.05} this rate is assumed to be
constant and the temporary impact of individual trades is estimated
proportional to $v^{0.6}$ which yields per-share temporary price impact
proportional to $Z(0)^{0.6}$. Here, in contrast, the temporary impact is
linear, proportional to $v$, but the \emph{optimal} rate of trading is
non-linear, roughly proportional to $\sqrt{Z}$ for small\ values. `Small'\
must be understood in context; we find that $\sqrt{Z}$ asymptotics is
perfectly compatible with meta-orders whose optimal execution lasts several
days, see Section 8.4.

We can also make qualitative conclusions about the optimized implementation
shortfall by studying the asymptoptic expansion (\ref{eq: u(x) series})
whereby we find that for small $Z(0)$ the per-share price impact equals 
\begin{equation*}
I(S(0),Z(0))=\frac{S(0)Z(0)-V(S(0),Z(0))}{S(0)Z(0)}=\frac{4}{3}\sqrt{\eta
(\rho -\lambda -r)Z(0)/S(0)}+O(Z(0)^{3/2}),
\end{equation*}%
which means that the price impact is proportional to the square root of the
total trade size. There is a strong empirical evidence to support the square
root law for meta-orders, see \cite{bershova.rakhlin.13}, \cite{farmer.al.13}%
, \cite{donier.al.15} and \cite{toth.al.16} and references therein.

\section{Computation of the solution}

\label{sect: comp}To make BVP$_{[0,\infty )}$ amenable to numerical
treatment we first truncate the spatial interval to $x\in \lbrack
\varepsilon ,L]$ with $\varepsilon \geq 0,L<\infty $ and solve the ODE (\ref%
{0}) with mixed boundary conditions $u(\varepsilon )=0$ and $u^{\prime
}(L)=0 $. We refer to the truncated boundary value problem as BVP$%
_{[\varepsilon ,L]}$. In section \ref{sect: BVP_L} we prove that the
solution $u_{L}$ of BVP$_{[0,L]}$ is unique and that it converges pointwise
upwards to the desired solution $u_{\infty }$ as $L\nearrow \infty $.

Numerical solutions of BVPs for ordinary differential equations with
singular coefficients have a well established literature, see for example 
\cite{jamet.70}, \cite{weinmuller.84}, \cite{weinmuller.86}, and \cite%
{auzinger.al.99} who consider BVPs with ODE of the form%
\begin{equation}
u^{\prime \prime }=x^{-1}A(x)u^{\prime }+x^{-2}B(x)u+F(x,u,u^{\prime }),
\label{eq: Weinmuller ODE}
\end{equation}%
where $A,B$ and $F$ are continuous at $x=0$ and one of the boundaries is $%
x=0 $. Numerical solution of (\ref{eq: Weinmuller ODE}) can be computed by
means of the Matlab function \textsf{bvp5c} after transformation $y(x)=[u(x)$%
\quad $xu^{\prime }(x)]$, see \cite{weinmuller.86}, equation (2.1a).

However, as we have mentioned already in the connection with IVP$_{0}$, our
problem BVP$_{[0,L]}$ is substantially more singular. This is not due to the
singularity in the linear terms of ODE (\ref{0}), which in fact can be
accommodated in the ansatz (\ref{eq: Weinmuller ODE}), but because the
non-linear part $F(x,u,u^{\prime })=\frac{1}{2}x^{-2}(u^{\prime }-1)^{2}$ is
not continuous in $x$ at zero. Attempts to compute the solution of BVP$%
_{[0,L]}$ by some kind of shooting fail -- both at $x\rightarrow 0$ and $%
x\rightarrow \infty $ the trajectories blow up. Algorithm \textsf{bvp5c }is
able to produce, with careful tuning of input parameters, a stable solution
of BVP$_{[\varepsilon ,L]}$ for $\varepsilon $ not too close to zero.
However, the quality of this solution near zero is poor, as can be seen in
panel (b) of Figure \ref{fig: 1}.

To bypass the troublesome singularity at zero we introduce a time dimension
into BVP$_{[0,L]}$ in a strategy akin to the \emph{value function iteration }%
method known from financial economics. This approach is also common in
linear-quadratic optimal control problems where, however, it is not
motivated by the presence of singularities, see 
\citet[Section 3.1]{anderson.moore.89}%
.

We consider a parabolic PDE that corresponds to a finite horizon version of
the time-homogeneous optimization (\ref{eq: (L)}). We formulate suitable
boundary conditions on a finite spatial interval $x\in \lbrack 0,L]$ to
obtain a parabolic problem BVP$_{[0,L]}^{t}$ and show that its solution
converges monotonically to the solution of BVP$_{[0,L]}$ as $t\rightarrow
\infty $. This is done in section \ref{sect: BVP_Lt}. Unfortunately, BVP$%
_{[0,L]}^{t}$ does\emph{\ }not\emph{\ }correspond to an optimal control
problem due to the choice of boundary conditions.

In section \ref{sect: finite difference} we formulate a finite difference
scheme to solve BVP$_{[0,L]}^{t}$ numerically. This scheme is well behaved
with respect to the singularity at $x=0$ and produces a reliable
approximation to $u_{L}$, which for large enough $L$ is arbitrarily close to
the desired solution $u_{\infty }$.

\subsection{Problem BVP$_{[0,L]}$\label{sect: BVP_L}}

\begin{theorem}
\label{truncation} Let $a+b>0$. For given $L>0$ $\mathrm{BVP}_{[0,L]}$ has a
unique solution $u_{L}\in C^{2}((0,L])\cap C^{0}([0,L])$ such that $0\leq
u_{L}(x)\leq x$ for all $x\in \lbrack 0,L]$. The solution $u_{L}$ is
strictly increasing, concave and satisfies $u_{L_{1}}(x)\leq u_{L_{2}}(x)$
for $L_{1}\leq L_{2},\ 0\leq x\leq L_{1}$, and $\lim_{L\rightarrow \infty
}u_{L}(x)=u_{\infty }(x)$ for $0\leq x<\infty $, where $u_{\infty }$ is the
unique solution of $\mathrm{BVP}_{[0,\infty )}$.
\end{theorem}

\begin{proof}
\emph{Step 1)} For any $\varepsilon >0$ such that $\varepsilon <L$ the
function $\alpha (x):=0$, resp. $\beta (x):=x$ is a lower (resp. upper)
solution of BVP$_{[\varepsilon ,L]}$ in the sense of Definition II.1.1 in 
\cite{decoster.habets.06}, which crucially allows for the Neumann boundary
condition at $L$. Therefore by Theorem II.1.3 \emph{ibid} the solution $%
u_{\varepsilon }$ of the mixed boundary value problem BVP$_{[\varepsilon
,L]} $ satisfies 
\begin{equation}
0\leq u_{\varepsilon }(x)\leq x\text{ for every }\varepsilon >0.
\label{eq: apriori u bound}
\end{equation}%
From here the proof proceeds as in Proposition 2.2 of \cite{brunovsky.al.13}%
. From Bernstein's condition \cite{bernstein.04} (see also Section I.4.3 of 
\cite{decoster.habets.06} for related Nagumo condition) fixing $\tilde{%
\varepsilon}>0$ we obtain a uniform (in $\varepsilon $) a-priori bound on
the derivative $u_{\varepsilon }^{\prime }$ on $[\tilde{\varepsilon},L]$.
Together with (\ref{eq: apriori u bound}) this yields via (\ref{0}) an
a-priori bound on $u_{\varepsilon }^{\prime \prime }$ on $[\tilde{\varepsilon%
},L]$ which means $\{u_{\varepsilon }^{\prime }\}_{\varepsilon >0}$ (as well
as $\{u_{\varepsilon }\}_{\varepsilon >0}$) are equicontinuous on $[\tilde{%
\varepsilon},L]$ which in turn implies equicontinuity of $\{u_{\varepsilon
}^{\prime \prime }\}_{\varepsilon >0}$ via (\ref{0}). One can thus extract a
convergent subsequence of $u_{1/k}$ which convergences with its first two
derivatives to some function $u$ on $(0,L]$ with $u(0)=0$ and such that $u$
solves (\ref{0}).

\emph{Step 2)} By \cite{brunovsky.al.13}, Lemma 3.1, $u_{L}^{\prime }(0)=1$.
This, together with the conditions $0\leq u_{L}(x)\leq x$ and $u_{L}^{\prime
}(L)=0$ excludes all alternatives of Proposition \ref{prop: alternatives}
except for ii). Therefore any solution of BVP$_{[0,L]}$ must be concave and
increasing on $[0,L]$.

\emph{Step 3)} To prove uniqueness of the solution assume that $u$ and $v$
are two solutions of BVP$_{[0.L]}$. Then $p:=v-u$ solves 
\begin{equation}
x^{2}p^{\prime \prime }=axp^{\prime }+bp-p^{\prime }(u^{\prime }-1)-\frac{1}{%
2}\left( p^{\prime }\right) ^{2},  \label{HJB_w}
\end{equation}%
on $(0,L)$ which on differentiation yields 
\begin{equation}
x^{2}p^{\prime \prime \prime }=\left( (a-2)x+1-u^{\prime }-p^{\prime
}\right) p^{\prime \prime }+(a+b-u^{\prime \prime })p^{\prime }.
\label{diff_w}
\end{equation}%
Applying Lemma 4.1 of \cite{brunovsky.al.13} to (\ref{diff_w}) with $%
y=p^{\prime }$, $g(x,y)=(a+b-u^{\prime \prime }(x))y$ and $y^{\ast }=0$, one
obtains that $p$ obeys the same alternatives as $u$ in Proposition \ref%
{prop: alternatives}.

By construction we have $p(0)=p^{\prime }(0)=p^{\prime }(L)=0$, therefore
alternatives (ii)-(v) of Proposition \ref{prop: alternatives} are excluded
and $p$ must be constant and thus necessarily equal to zero. Thus BVP$%
_{[0,L]}$ has a unique solution which we denote by $u_{L}$.

\emph{Step 4)}\ Now we prove that the solutions $u_{L}$ grow with $L$. Take $%
0<L<K$ and let $u:=u_{L}$, $v:=u_{K}$. Consider $p:=v-u$ on $(0,L)$ which
satisfies (\ref{HJB_w}), (\ref{diff_w}) and therefore obeys the alternatives
of Proposition \ref{prop: alternatives}.. As before we have $p^{\prime
}(0)=0 $. Since $v^{\prime }(L)>0$ while $u^{\prime }(L)=0$ we also have $%
p^{\prime }(L)>0.$ Hence in Proposition \ref{prop: alternatives} (iii) is
the only possible alternative, $p$ is strictly convex on $(0,L)$ and
therefore $p^{\prime }>0$ on $(0,L]$ which implies $u_{K}^{\prime
}>u_{L}^{\prime }$ and $u_{K}>u_{L}$ on $(0,L]$.

\emph{Step 5)}\ It remains to be proved that for $L\rightarrow \infty $, $%
u_{L}$ converges pointwise to the solution of BVP$_{[0,\infty )}$. Step 2)\
implies $0\leq u_{L}(x)\leq x$ and by step 4)\ $u_{L}(x)$ is increasing in $%
L $ therefore for fixed $x$ the limit $\lim_{L\rightarrow \infty }u_{L}(x)=:%
\tilde{u}(x)$ is well defined. Likewise $0\leq u_{L}^{\prime }(x)\leq 1$ and 
$u_{L}^{\prime }$ is increasing in $L$ hence we have a well-defined limit $%
\lim_{L\rightarrow \infty }u_{L}^{\prime }(x)=:\tilde{v}(x)$. Picking
arbitrary $x$ and $x_{0}$ in $(0,\infty )$ we rewrite \eqref{0} in integral
form 
\begin{align}
u_{L}(x)& =u_{L}(x_{0})+\int_{x_{0}}^{x}u_{L}^{\prime }(\xi )\mathrm{d}\xi ,
\label{u1_int} \\
u_{L}^{\prime }(x)& =u_{L}^{\prime }(x_{0})+\int_{x_{0}}^{x}f\Big(\xi
,u_{L}(\xi ),u_{L}^{\prime }(\xi )\Big)\mathrm{d}\xi  \label{u2_int}
\end{align}%
with 
\begin{equation}
f(x,u,v)=a\,\frac{v}{x}+b\,\frac{u}{x^{2}}-\frac{1}{2}\,\frac{(v-1)^{2}}{%
x^{2}}.  \label{def_f}
\end{equation}%
Passing to the limit $L\rightarrow \infty $ in (\ref{u1_int}, \ref{u2_int})
and using dominated convergence yields%
\begin{align*}
\tilde{u}(x)& =\tilde{u}(x_{0})+\int_{x_{0}}^{x}\tilde{v}(\xi )\mathrm{d}\xi
, \\
\tilde{v}(x)& =\tilde{v}(x_{0})+\int_{x_{0}}^{x}f\Big(\xi ,\tilde{u}(\xi ),%
\tilde{v}(\xi )\Big)\mathrm{d}\xi ,
\end{align*}%
which on differentiation shows that $\tilde{u}$ solves ODE (\ref{0}) on $%
(0,\infty )$. Since $0\leq \tilde{u}(x)\leq x,$ by Propositions \ref{bcw}
and \ref{prop: bvp_infty} $\tilde{u}$ solves BVP$_{[0,\infty )}$.
\end{proof}

\subsection{BVP$_{[0,L]}$ as a limit of finite horizon problems BVP$%
_{[0,L]}^{t}$\label{sect: BVP_Lt}}

At this point the singularity of BVP$_{[0,L]}$ at zero is still a major
obstacle in obtaining a reliable numerical solution. To bypass the
singularity we will consider a parabolic PDE generated by the ODE (\ref{0}), 
\begin{equation}
w_{t}=x^{2}w_{xx}-axw_{x}-bw+\frac{1}{2}(w_{x}-1)^{2},  \label{PDE}
\end{equation}%
with the boundary conditions 
\begin{eqnarray}
w(t,\varepsilon ) &=&0,  \label{lokraj} \\
w_{x}(t,L) &=&0,  \label{pokraj}
\end{eqnarray}%
and initial condition 
\begin{equation}
w(0,x)=0.  \label{zdola}
\end{equation}%
We refer to the boundary value problem (\ref{PDE}-\ref{zdola}) on $[0,\infty
)\times \lbrack \varepsilon ,L]$ as $\underline{\mathrm{BVP}}_{[\varepsilon
,L]}^{t}.$ When the initial condition (\ref{zdola}) is replaced with 
\begin{equation}
w(0,x)=x,  \label{zhora}
\end{equation}%
we speak of $\overline{\mathrm{BVP}}_{[\varepsilon ,L]}^{t}$.

Three related difficulties have to be mastered. First, the parabolicity of
PDE (\ref{PDE}) degenerates at $x=0$, so basic theory of semilinear
parabolic equations is not applicable directly. Second, the truncation to
finite spatial interval breaks the link between the BVP and the optimal
control problem (\ref{eq: (L)}), so we cannot appeal to results from optimal
control literature. Third, standard existence theorems do not cover mixed
boundary conditions (Dirichlet on the left, Neumann on the right) since most
of this\ theory is developed in higher dimensions where boundary is a
connected set. We prove,

\begin{theorem}
\label{theo: BVP Lt} For given $L$ the problems $\underline{\mathrm{BVP}}%
_{[0,L]}^{t}$ and $\overline{\mathrm{BVP}}_{[0,L]}^{t}$ have a unique
solution in $\mathcal{C}^{1,2}((0,\infty )\times (0,L])\cap \mathcal{C}%
([0,\infty )\times \lbrack 0,L])$. These solutions, denoted by $\underline{w}
$ and $\overline{w}$ respectively, satisfy 
\begin{eqnarray}
0 &\leq &\underline{w}(t,x)\leq u_{L}(x)\leq \overline{w}(t,x)\leq x,
\label{horedolu0} \\
\frac{\partial \overline{w}(t,x)}{\partial t} &\leq &0\leq \frac{\partial 
\underline{w}(t,x)}{\partial t},  \label{eq: w_L upwards}
\end{eqnarray}%
and $\lim_{t\rightarrow \infty }\overline{w}(t,x)=\lim_{t\rightarrow \infty }%
\underline{w}(t,x)=u_{L}(x).$
\end{theorem}

{We only spell out the proof for $\underline{\mathrm{BVP}}_{[0,L]}^{t}$, the
other case being analogous. }We tackle the proof by studying a spatially
symmetric version of $\underline{\mathrm{BVP}}_{[\varepsilon ,L]}^{t}$ on
the interval $[\varepsilon ,2L-\varepsilon ],$ denoted by $\underline{%
\mathrm{SBVP}}_{[\varepsilon ,2L-\varepsilon ]}^{t}.$ The symmetric problem
has boundary conditions of Dirichlet type at both ends which allows us to
refer to the literature more comfortably. Moreover, $L$ is in the interior
of the spatial domain of the symmetric problem, and this gives us access to
uniform a-priori estimates of the spatial derivative near $L$, making the
limiting procedure for $\varepsilon \rightarrow 0$ less involved. The
conclusions of Theorem \ref{theo: BVP Lt} become a simple corollary of the
results for $\underline{\mathrm{SBVP}}_{[0,2L]}^{t}$. The price we have to
pay for taking the symmetrization route is discontinuity of coefficients at $%
x=L$.

\begin{definition}
A function $w^{\varepsilon }\in C^{1,2}((0,\infty )\times (\varepsilon
,2L-\varepsilon ))\cap C([0,\infty )\times \lbrack \varepsilon
,2L-\varepsilon ])$ is said to be a solution of $\underline{\mathrm{SBVP}}%
_{[\varepsilon ,2L-\varepsilon ]}^{t}$, if i) it is symmetric with respect
to $L,$ i.e. $w^{\varepsilon }(t,x)=w^{\varepsilon }(t,2L-x)$; ii) it
satisfies 
\begin{equation}
w_{t}^{\varepsilon }=M(x)w_{xx}^{\varepsilon }-A(x)w_{x}^{\varepsilon
}-bw^{\varepsilon }+C(x,w_{x}^{\varepsilon })  \label{S}
\end{equation}%
on $(0,\infty )\times (\varepsilon ,2L-\varepsilon )$, (\ref{lokraj}), and (%
\ref{zdola}) for $x\in \lbrack \varepsilon ,2L-\varepsilon ],$ where 
\begin{eqnarray*}
M(x) &=&%
\begin{cases}
{x^{2}} & \mbox{ for }0\leq x\leq L \\ 
{(2L-x)^{2}} & \mbox{ for }L\leq x\leq 2L%
\end{cases}
\\
A(x) &=&%
\begin{cases}
ax & \mbox{ for }0\leq x\leq L \\ 
-a(2L-x) & \mbox{ for }L<x\leq 2L.%
\end{cases}
\\
C(x,p) &=&\frac{1}{2}(\emph{sign}(L-x)p-1)^{2}
\end{eqnarray*}
\end{definition}

\begin{remark}
\label{discontinuity} Function $A$ is discontinuous at $x=L$. The same is
true of $C(x,p)$ unless $p=0$. In what follows we will employ a-priori
estimates from \cite{lieberman.96}, \cite{ladyzhenskaya.al.68} that
ostensibly assume continuity of the data of the equation. Nevertheless, a
close inspection of the arguments reveals that one only needs continuity of
the terms obtained by composition of the data with the solutions, that is
continuity of $M(x)w_{xx}^{\varepsilon }$, $A(x)w_{x}^{\varepsilon }$, and $%
C(x,w_{x}^{\varepsilon })$. This holds true in our case because any smooth
spatially symmetric function $w^{\varepsilon }(t,x)$ has $w_{x}^{\varepsilon
}(t,L)=0$.
\end{remark}

To establish existence and uniqueness of solutions to $\underline{\mathrm{%
SBVP}}_{[\varepsilon ,2L-\varepsilon ]}^{t}$ for $\varepsilon >0$ we apply
the theory of analytic semigroups \cite{henry.81}.

\begin{lemma}
\label{regexist} For given $0<\varepsilon <L$, $\underline{\mathrm{SBVP}}%
_{[\varepsilon ,2L-\varepsilon ]}^{t}$ has a unique solution $w^{\varepsilon
}$ satisfying 
\begin{equation}
0\leq w^{\varepsilon }(t,x)\leq \min \{x,2L-x\}\text{ on }[0,\infty )\times
\lbrack \varepsilon ,2L-\varepsilon ],  \label{horedolu}
\end{equation}%
and for $0<\varepsilon _{1}<\varepsilon _{2}<L$ 
\begin{equation}
w^{\varepsilon _{1}}\geq w^{\varepsilon _{2}}\text{ on }[0,\infty )\times
\lbrack \varepsilon _{2},2L-\varepsilon _{2}].  \label{dohora}
\end{equation}
\end{lemma}

\begin{proof}
Denote $X=L_{2}(\varepsilon ,2L-\varepsilon )\cap \{y:y(x)=y(2L-x)\}$.
Further, define $\mathcal{M}:D(\mathcal{M})=X\cap H_{0}^{1}(\varepsilon
,2L-\varepsilon )\cap H^{2}(\varepsilon ,2L-\varepsilon )\rightarrow X$ by 
\begin{equation*}
(\mathcal{M}y)(x)=-M(x)y^{\prime \prime }(x)
\end{equation*}%
$\mathcal{M}$ is a linear unbounded densely defined operator $D(\mathcal{M}%
)\rightarrow X$. From the Sturm-Liouville theory of linear boundary value
problems for second order linear ordinary differential equations it follows
that the spectrum of $\mathcal{M}$ consists of a sequence of real
eigenvalues with the only accumulation point $\infty $. Consequently, $%
\mathcal{M}$ is sectorial (\cite{henry.81}, Definition 1.3.1) and, thus, the
infinitesimal generator of an analytic semigroup (\cite{henry.81},
Definition 1.3.3). As such, it admits the fractional power $\mathcal{M}^{1/2}
$ (\cite{henry.81}, Definition 1.4.1) which is a densely defined linear
operator $D(\mathcal{M}^{1/2})\rightarrow X$, $X^{1/2}=D(\mathcal{M}%
^{1/2})\in X$ (\cite{henry.81}, Definition 1.4.7). For our $\mathcal{M}$ one
has $X^{1/2}=H_{0}^{1}(0,2L)$,{\ which is by definition the space of
functions vanishing on the set $\{0,2L\}$ with derivatives in $L_{2}(0,2L)$}
(\cite{henry.81}, Example{\ 6 of Section 1.4}).

Following \cite{henry.81} we write our problem as an abstract differential
equation 
\begin{equation}
dy/dt+\mathcal{M}y=f(y)  \label{henry}
\end{equation}%
for $y\in X$ and $f:X^{1/2}\mapsto X$ given by 
\begin{equation*}
f(y)(x)=-A(x)y^{\prime }(x)-by(x)+C(x,y{\color{red}^{\prime }}(x)).
\end{equation*}%
Since $f$ is locally Lipschitz continuous, local existence and uniqueness of
the solution of the problem (\ref{henry}), $y(0)=0$, is provided by \cite%
{henry.81}, Theorem 3.3.3.

Inequality (\ref{horedolu}) follows from the fact that $0$ is a subsolution
and $\min \{x,2L-x\}$ is a supersolution of the problem $\underline{\mathrm{%
SBVP}}_{[\varepsilon ,2L-\varepsilon ]}^{t}$. From \cite{lieberman.96},
Theorem 10.17 it follows that $w_{x}^{\varepsilon }$ is bounded as well, the
bound depending only on the bound of $w^{\varepsilon }$. That is, the local
solution $y(t)$ is bounded in $X^{1/2}=H_{0}^{1}$. From \cite{henry.81},
Theorem 3.3.4 is thus follows that the solution extends to $t\in \lbrack
0,\infty )$. The inequality (\ref{dohora}) follows similarly, since the
function $w^{\varepsilon _{2}}$ extended by 0 to $[0,\infty )\times \lbrack
\varepsilon _{1},\varepsilon _{2}]\cup \lbrack 2L-\varepsilon
_{2},2L-\varepsilon _{1}]$ is a subsolution for $\underline{\mathrm{SBVP}}%
_{[\varepsilon _{1},2L-\varepsilon _{1}]}^{t}$.
\end{proof}

We now describe the limiting procedure for $\varepsilon \rightarrow 0.$

\begin{proposition}
\label{prop: BVP Lt} For given $L$ the problem $\underline{\mathrm{SBVP}}%
_{[0,2L]}^{t}$ has a unique solution $w\in \mathcal{C}^{1,2}((0,\infty
)\times (0,2L)\cap \mathcal{C}([0,\infty )\times \lbrack 0,2L])$. This
solution satisfies 
\begin{eqnarray}
0 &\leq &w(t,x)\leq \min \{x,2L-x\},  \label{horedolu0_s} \\
\frac{\partial w(t,x)}{\partial t} &\geq &0.  \label{eq: w_L upwards_s}
\end{eqnarray}
\end{proposition}

\begin{proof}
\emph{Step 1)} Denote by $w^{\varepsilon }$ the unique solution of $%
\underline{\mathrm{SBVP}}_{[\varepsilon ,2L-\varepsilon ]}^{t}$. By Lemma %
\ref{regexist} the family of functions $w^{\varepsilon }$ is bounded from
above and increasing as $\varepsilon \searrow 0$ . Hence it has a pointwise
limit $w$ which satisfies (\ref{horedolu0_s}) thanks to (\ref{horedolu}).
Trivially, $w(t,x)=w(t,2L-x)$ and $w(t,0)=0$. We will show that $w$ is in
fact a solution of $\underline{\mathrm{SBVP}}_{[0,2L]}^{t}$.

\emph{Step 2)} Choose $\varepsilon <x_{1}<x_{2}<2L-\varepsilon ,\ 0<\tau <T$
and denote $G=(\tau ,T)\times (x_{1},x_{2})$. Because the nonlinear term $C$
satisfies the Bernstein condition of quadratic growth, by Theorem 12.2 of 
\cite{lieberman.96}, the functions $w_{x}^{\varepsilon }$ are uniformly H%
\"{o}lder continuous in $G$. Therefore, we can find a sequence $\varepsilon
_{n}\rightarrow 0$ such that both $w^{\varepsilon _{n}}$ and $%
w_{x}^{\varepsilon _{n}}$ converge uniformly in $G$ to $w,w_{x}$,
respectively.

\emph{Step 3)} We will now show that $w$ is a weak solution of PDE %
\eqref{PDE} on $G$. Take any function $\phi \in \mathcal{C}^{\infty }(%
\overline{G})$ which vanishes with all its derivatives at the boundary of $G$
and $n$ so large that $[0,\infty )\times \lbrack \varepsilon _{n},L]\supset
G $. Since $w^{\varepsilon _{n}}$ solves \eqref{PDE} in $G$, one has 
\begin{equation*}
\int_{G}[w_{t}^{\varepsilon _{n}}-M(x)w_{xx}^{\varepsilon
_{n}}+A(x)w_{x}^{\varepsilon _{n}}+bw^{\varepsilon
_{n}}-C(x,w_{x}^{\varepsilon _{n}})]\phi dtdx=0,
\end{equation*}%
or equivalently, 
\begin{equation*}
\int_{G}[(w_{t}^{\varepsilon _{n}}-(M(x)w_{x}^{\varepsilon
_{n}})_{x}+N(x,w^{\varepsilon },w_{x}^{\varepsilon })]\phi dtdx=0
\end{equation*}%
where 
\begin{equation*}
N(x,w,p)=%
\begin{cases}
(-2+a)xp+bw-\frac{1}{2}(p-1)^{2} & \mbox{ for }0\leq x\leq L \\ 
(2-a)(2L-x)p+bw-\frac{1}{2}(-p-1)^{2} & \mbox{ for }L<x\leq 2L.%
\end{cases}%
\end{equation*}%
Integrating the first two terms by parts we obtain 
\begin{equation*}
-\int_{G}w^{\varepsilon _{n}}\phi _{t}dxdt+\int_{G}M(x)w_{x}^{\varepsilon
_{n}}\phi _{x}dxdt+\int_{G}N(x,w^{\varepsilon },w_{x}^{\varepsilon })\phi
dtdx=0.
\end{equation*}%
Because of uniform convergence of the sequences $\{w^{\varepsilon
_{n}}\}_{n} $ and $\{w_{x}^{\varepsilon _{n}}\}_{n}$ we can pass to the
limit to obtain 
\begin{equation*}
-\int_{G}w\phi _{t}dxdt+\int_{G}w_{x}M(x)\phi
_{x}dxdt+\int_{G}N(x,w,w_{x})\phi dtdx=0.
\end{equation*}

\emph{Step 4)} Since both $0<x_{1}<x_{2}<2L$, $0<\tau <T$ and $\phi $ are
arbitrary this means that $w$ is a weak solution and consequently, a
classical solution as well on any interior subdomain (\cite%
{ladyzhenskaya.al.68},VI.1). As such, it is $C^{1,2}((0,\infty )\times
(0,2L) $.

\emph{Step 5)} Since the functions $w^{\varepsilon }$ satisfy \eqref{zdola},
to prove that $w$ satisfies \eqref{zdola} as well, it suffices to prove that
for fixed $x_{0}\in (0,L)$, $w$ is equicontinuous on $t$, uniformly with
respect to $\varepsilon $ and $x\in \lbrack x_{1},x_{2}],t\in \lbrack 0,T],\
0<x_{1}<x<x_{2}<L,\ T>0$. This, however, follows from \cite%
{ladyzhenskaya.al.68}, Theorem V.3.1, according to which $\Vert
w_{t}^{\varepsilon }\Vert _{L_{2}[0,T]}$ is bounded uniformly with respect
to $(t,x)\in \lbrack 0,T]\times \lbrack x_{1},x_{2}]$ and $\varepsilon >0$.

\emph{Step 6)} {Uniqueness of the solution follows from the parabolic
maximum principle \cite{lieberman.96}, Theorem 2.10, applied to the
difference of solutions. }

\emph{Step 7)} In a straightforward way one can verify that function $v=w_{t}
$ is a weak solution of the problem 
\begin{eqnarray}
v_{t} &=&M(x)v_{xx}-bv-(A(x)-\hat{C}(t,x))v_x  \label{rovder}
\\
v(t,0) &=&0,\ v(t,2L)=0,\ v(0,x)=\frac{1}{2};  \label{ostatne}
\end{eqnarray}%
where 
\begin{equation*}
\hat{C}(t,x)=%
\begin{cases}
w_{x}(t,x)-1 & \mbox{ for }0\leq x\leq L \\ 
{w_{x}(t,x)+1} & \mbox{ for }L<x\leq 2L;%
\end{cases}%
\end{equation*}%
the initial condition for $v$ following from \eqref{S} following by
substitution of $w(t,0)=0$ into \eqref{S}. By \cite{ladyzhenskaya.al.68},
VI.2 and Remark \ref{discontinuity} $v$ is a classical solution. Since 0 is
a subsolution of the problem (\ref{rovder}), (\ref{ostatne}), its solution $%
v=w_{t}$ is nonnegative.
\end{proof}

Finally, we prove convergence for $t\rightarrow \infty $.

\begin{proposition}
\label{convergence} For $t\rightarrow \infty $ the solution of the problem $%
\underline{\mathrm{SBVP}}_{[0,2L]}^{t}$ converges to a (stationary) solution
of $\mathrm{SBVP}_{[0,2L]}$, defined as time-independent solution of $%
\underline{\mathrm{SBVP}}_{[0,2L]}^{t}$ without the boundary condition (\ref%
{lokraj}).
\end{proposition}

\begin{proof}
\emph{Step 1)} Since the solution $w$ of $\underline{\mathrm{SBVP}}%
_{[0,2L]}^{t}$ is increasing in $t$ and bounded by Proposition \ref{prop:
BVP Lt}, for $t\rightarrow \infty $ it converges pointwise to a function $u$
on $[0,2L]$ satisfying 
\begin{equation}
0\leq u(x)\leq \min \{x,2L-x\}.  \label{Uhoredolu}
\end{equation}%
We wish to show that $u$ solves $\mathrm{SBVP}_{[0,2L]}$.

\emph{Step 2) } From \cite{lieberman.96}, Theorem 12.2 it follows that for
any fixed $0<l<L,\ T>0$, $w_{x}$ is bounded on $(T,\infty )\times \lbrack
l,2L-l]$. Therefore, the family of functions $w(t,\cdot )$ is equicontinuous
on $[l,2L-l]$. Because by (\ref{horedolu0_s}) it is uniformly bounded, its
convergence to $u$ on $[l,2L-l]$ is uniform. Consequently, $u$ is continuous
on $(0,2L)$. Because of (\ref{horedolu0_s}) its continuity extends to $%
[0,2L] $.

\emph{Step 3) }By \cite{lieberman.96}, Theorems 12.25 and 12.2, for fixed $l$%
, the problem 
\begin{eqnarray}
W_{t} &=&M(x)W_{xx}-A(x)W_{x}-bW+C(x,W_{x}))\mbox{ for }l\leq x\leq 2L-l
\label{WPDE} \\
&&W(0,x)=u(x),\ W(t,l)=W(t,2L-l)=u(l)  \label{Wostatne}
\end{eqnarray}%
has a unique solution $W\in C^{1,2}((0,\infty )\times (l,2L-l))\cap
C^{0}([0,\infty )\times \lbrack l,2L-l])$ and, for fixed $\tau >0$, $W_{x}$
is bounded on $[\tau ,\infty )$. We wish to show that $W(t,x)\equiv u(x)$
for each $l$ which immediately implies that $u$ solves $\mathrm{SBVP}%
_{[0,2L]}$.

Fix $\tau ,T>0$ and for $0\leq t\leq \tau ,\ l\leq x\leq 2L-l$ denote 
\begin{equation}
Y^{T}(t,x)=W(t,x)-w(T+t,x).  \label{delta}
\end{equation}%
The function $Y^{T}$ solves the linear problem%
\begin{eqnarray}
Y_{t}^{T} &=&M(x)Y_{xx}^{T}-(A(x)-Q(t,x))Y_{x}^{T}-bY^{T}  \label{DPDE} \\
0 &\leq &Y^{T}(0,x)={u}(x)-w(T,x)\leq \varepsilon (T) \\
0 &\leq &Y^{T}(t,l)={u}(l)-w(T+t,l)\leq \varepsilon (T) \\
0 &\leq &Y^{T}(t,2L-l)={u}(l)-w(T+t,2L-l)\leq \varepsilon (T),
\label{Dostatne}
\end{eqnarray}%
where 
\begin{equation*}
Q(t,x)={%
\begin{cases}
\frac{1}{2}(W_{x}(t,x)+w_{x}(t,x)-2) & \mbox{ for }0\leq x\leq L \\ 
\frac{1}{2}(W_{x}(t,x)+w_{x}(t,x)+2) & \mbox{ for }L<x\leq 2L,%
\end{cases}%
}
\end{equation*}%
and $\varepsilon (T)\rightarrow 0$ for $T\rightarrow \infty $. For fixed $%
\tau >0$, $w_{x}(T+t,x),\ W_{x}(t,x)$ are both uniformly bounded for $0\leq
t\leq \tau ,l\leq x\leq L-l$ and so are $M,N$. Let $\beta $ be the uniform
bound of $M$. By the maximum principle for parabolic PDE (\cite{lieberman.96}%
, Theorem 2.4), one obtains $0\leq Y^{T}(t,x)\leq e^{\beta \tau }\varepsilon
(T),$ or equivalently, 
\begin{equation*}
W(t,x)=\lim_{T\rightarrow \infty }w(T+t,x)=u(x)\mbox{ for all }0\leq t\leq
\tau .
\end{equation*}
\end{proof}

\begin{proof}[Proof of Theorem \protect\ref{theo: BVP Lt}]
Let $\underline{w}$ be the unique solution of $\underline{\mathrm{SBVP}}%
_{[0,2L]}^{t}$ established in Proposition \ref{prop: BVP Lt}. Because of
symmetry its restriction $\underline{w}|_{[0,L]}$ solves $\underline{\mathrm{%
BVP}}_{[0,L]}^{t}$. Conversely, since the symmetric extension of any
solution of $\underline{\mathrm{BVP}}_{[0,L]}^{t}$ is a solution of $%
\underline{\mathrm{SBVP}}_{[0,2L]}^{t}$ and the latter is unique, $%
\underline{w}|_{[0,L]}$ is the unique solution of $\underline{\mathrm{BVP}}%
_{[0,L]}^{t}$. By Proposition \ref{convergence} $\underline{w}|_{[0,L]}$
converges to a stationary solution of $\underline{\mathrm{BVP}}_{[0,L]}^{t}$%
, i. e. to a solution of $\underline{\mathrm{BVP}}_{[0,L]}$ known to be
unique by Theorem \ref{truncation}.
\end{proof}

\subsection{Finite difference scheme for BVP$_{[0,L]}^{t}$ \label{sect:
finite difference}}

For the spatial variable $x$ we employ a non-equidistant partition defined
by $x_{j}=e^{\xi _{j}}-1-\xi _{j}+\xi _{j}^{3/2}$, $j=0,1,\dots ,N$, where
the points $\{\xi _{j}\}_{j=0}^{N}$ are equidistant, $x_{0}=0$ and $x_{N}=L$%
. We use a uniform time grid with $M$ points and step $h=T/M.$ In vector
notation the explicit finite difference scheme reads%
\begin{equation}
w_{i,1:N-1}=w_{i-1,1:(N-1)}+h\,\left( Aw_{i-1,\cdot }+F(w_{i-1,\cdot
})\right) \text{ for }i=1,\ldots ,M,  \label{euler_internal}
\end{equation}%
where the non-zero terms of matrix $A\in \mathbb{R}^{\left( N-1\right)
\times \left( N+1\right) }$ are given by 
\begin{align*}
A_{j,j-1}& =\frac{2\,x_{j}^{2}}{(x_{j+1}-x_{j-1})(x_{j}-x_{j-1})}+\frac{%
a\,x_{j}}{x_{j+1}-x_{j-1}}, \\
A_{j,j}& =-\frac{2\,x_{j}^{2}}{x_{j+1}-x_{j-1}}\left( \frac{1}{x_{j+1}-x_{j}}%
+\frac{1}{x_{j}-x_{j-1}}\right) -b, \\
A_{j,j+1}& =\frac{2\,x_{j}^{2}}{(x_{j+1}-x_{j-1})(x_{j+1}-x_{j})}-\frac{%
a\,x_{j}}{x_{j+1}-x_{j-1}},
\end{align*}%
for $j=1,2,\dots ,N-1$.

The non-linear term $F$ is given by 
\begin{equation*}
F(w_{i,\cdot })^{\top }=\frac{1}{2}\left[ 
\begin{array}{ccccc}
\left( \frac{w_{i,2}-w_{i,0}}{x_{2}-x_{0}}-1\right) ^{2} & \cdots & \left( 
\frac{w_{i,j+1}-w_{i,j-1}}{x_{j+1}-x_{j-1}}-1\right) ^{2} & \cdots & \left( 
\frac{w_{i,N}-w_{i,N-2}}{x_{N}-x_{N-2}}-1\right) ^{2}%
\end{array}%
\right] ,
\end{equation*}%
the boundary values are given by%
\begin{equation}
w_{i,0}=0,\qquad w_{i,N}=w_{i,N-1},  \label{euler_boundary}
\end{equation}%
and the initial condition is $w_{0,\cdot }=0$ for $\underline{\mathrm{BVP}}%
_{[0,L]}^{t}$ or $w_{0,\cdot }=x$ in the case of $\overline{\mathrm{BVP}}%
_{[0,L]}^{t}$.

Given $L$, $N$, time step $h$ and an initial condition for $w(0,x)$ we are
able to calculate an approximation of $w(t_{i+1},x)$ from the currently
known time layer $w(t_{i},x)$ using (\ref{euler_internal}) and (\ref%
{euler_boundary}). As proposed earlier the solutions of $\underline{\mathrm{%
BVP}}_{[0,L]}^{t}$ and $\overline{\mathrm{BVP}}_{[0,L]}^{t}$ converge
monotonically from below, resp. from above, to $u_{L}$, the solution of $%
\mathrm{BVP}_{[0,L]}$. Their convergence is demonstrated in panel (a) of
Figure \ref{fig: 1} and occurs numerically for $t=2$. In panel (b) we
contrast our solution with the one produced by Matlab solver \textsf{bvp5c}
designed to solve a less singular problem (\ref{eq: Weinmuller ODE}).

\begin{figure}[t]
\centering%
\begin{subfigure}[b]{0.49\textwidth}
        		\centering
               \includegraphics[height=5.2cm]{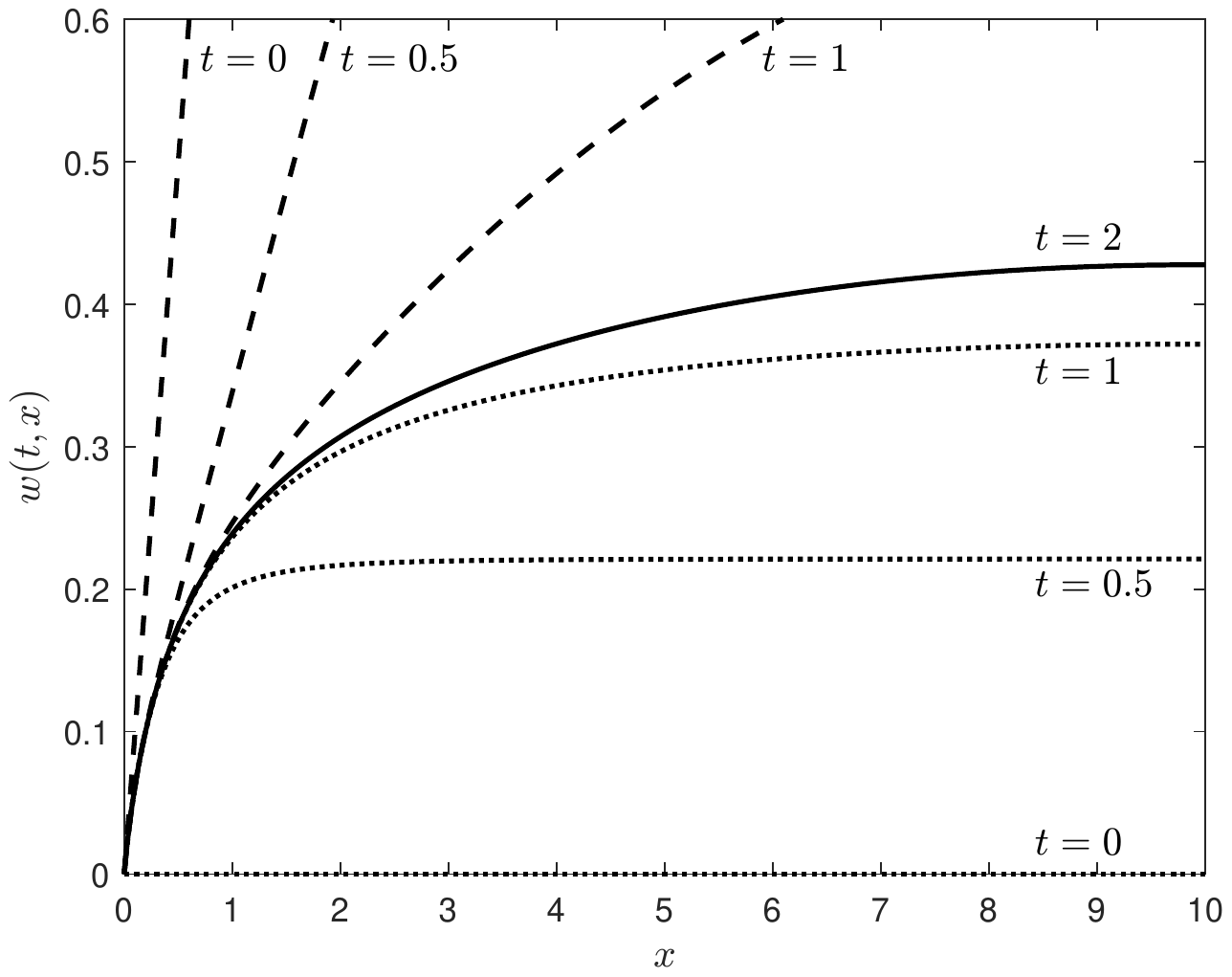}
                \caption{}
        \end{subfigure}%
\begin{subfigure}[b]{0.49\textwidth}
                \centering
                \includegraphics[height=5.2cm]{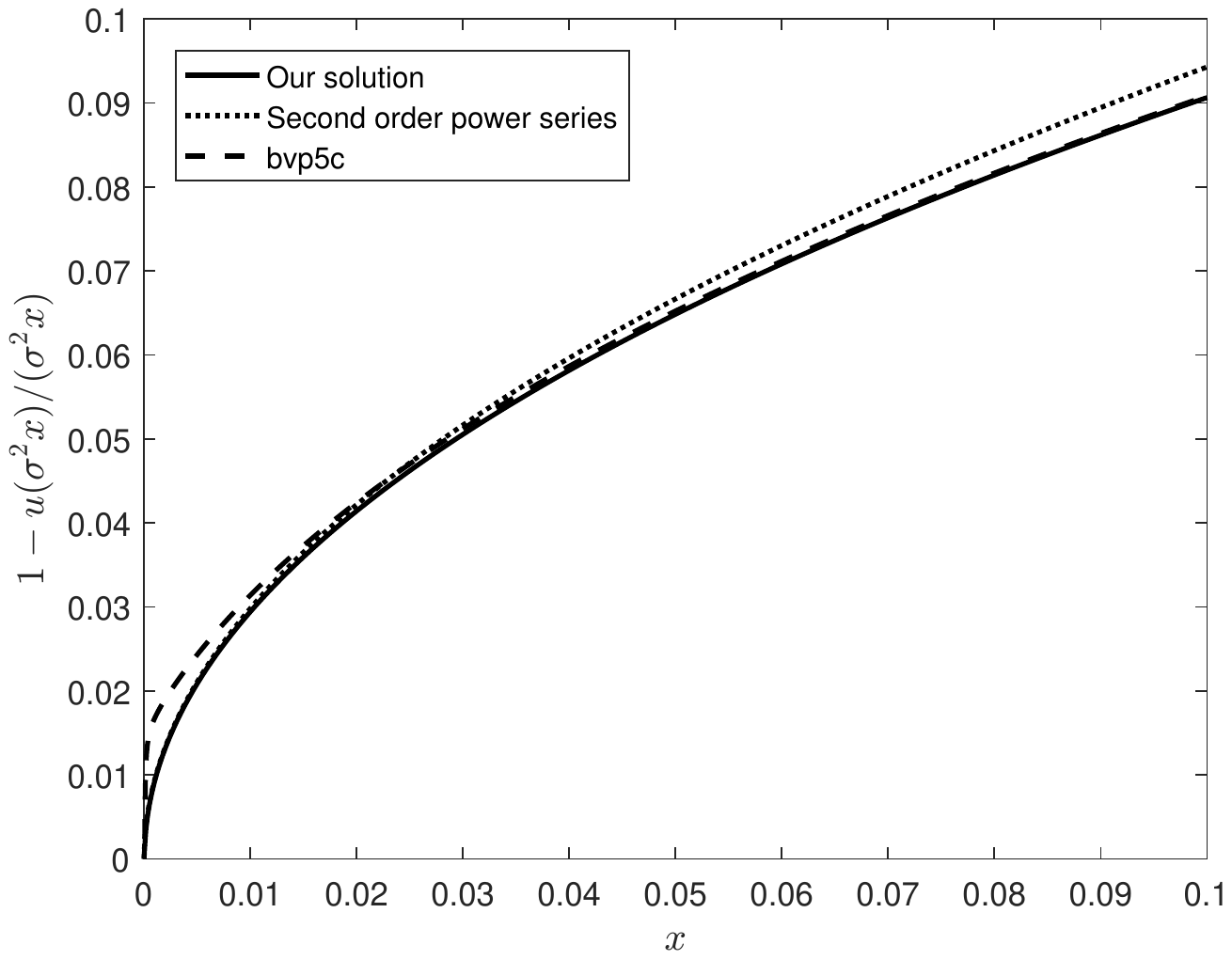}
                \caption{}
        \end{subfigure}
\caption{(a) Solutions of $\protect\underline{\mathrm{BVP}}_{[0,L]}^{t}$
(dotted) and $\overline{\mathrm{BVP}}_{[0,L]}^{t}$ (dashed) for $L=10$ and
different values of $t$. Solid line represents solution of $\mathrm{BVP}%
_{[0,L]}.$ (b) Comparison of $\mathrm{BVP}_{[0,L]}$ solution to solution
from Matlab routine \textsf{bvp5c}. The displayed quantity $1-u_{L}(\protect%
\sigma ^{2}x)/(\protect\sigma ^{2}x)$ represents approximate implementation
shortfall.}
\label{fig: 1}
\end{figure}

We aim to compute $u_{\infty }$ with sufficient precision on the interval $%
[0,1]$. The procedure has four nested loops. In the innermost loop, for a
chosen time step $h$, length of the spatial interval $L\geq 1$, and number
of partition points of the spatial interval $N\geq 10$ we determine the time
horizon $T$ (and thus also the number of time steps $M=T/h$) in the
following way. We consider two time layers, $T_{1}<T_{2}$ and the
corresponding numerical solutions $u_{i}(x):=w(T_{i},x)$ for $i=1,2$, which
we reparametrize in terms of relative implementation shortfall $%
f_{i}(x):=1-u_{i}(x)/x$. We distinguish between two regions for $x$: $%
\mathcal{X}=\left\{ x>0:f_{2}(x)\leq 0.01\right\} $ and its complement in $%
[0,1]$ denoted by $\mathcal{X}^{c}$.

For small $x$, we consider relative difference in $f_{i}$. Specifically, we
aim to attain 
\begin{equation}
\sup_{x\in \mathcal{X}}\left\vert 1-f_{2}(x)/f_{1}(x)\right\vert \leq 0.1.
\label{stop1}
\end{equation}%
For the remaining values of $x$ in the interval $[0,1]$ we target the
absolute difference in $f_{i}$ 
\begin{equation}
\sup_{x\in \mathcal{X}^{c}}\left\vert f_{2}(x)-f_{1}(x)\right\vert \leq
10^{-4}.  \label{stop2}
\end{equation}%
We start with $T_{1}=0.1$, $T_{2}=0.2$ and increase $T_{i}$ by $0.1$ until
conditions (\ref{stop1}) and (\ref{stop2}) are satisfied.

One level up, for given $L,h$ we start with $N_{1}=10,$ $N_{2}=20,$ denoting
the corresponding solutions obtained in the innermost loop by $u_{1}$ and $%
u_{2}$. We increase $N_{i}$ by $10$ until conditions (\ref{stop1}) and (\ref%
{stop2}) are met again.

Two levels up, for fixed $h$ we start with $L_{1}=1$ and $L_{2}=1.1.$ We
improve computational efficiency by using $u_{1}$ extended to the interval $%
[0,L_{2}]$ by a constant value, as the initial condition when computing $%
u_{2}$. We keep increasing $L_{i}$ by $0.1$ until conditions (\ref{stop1})
and (\ref{stop2}) are met.

In the outermost loop we check that the time step $h$ is sufficiently small
so as not to have any effect on the final solution. We start with $%
h_{1}=10^{-5}$ and $h_{2}=0.5\times 10^{-5}$ and denote corresponding
solutions determined by the previous loop by $u_{1}$ and $u_{2}$. We keep
halving the time step until conditions (\ref{stop1}) and (\ref{stop2}) are
met. Whenever possible we use previously computed values of $u$ as an
initial guess for the next step of the procedure. When passing from a
coarser to a finer mesh we perform this by cubic spline interpolation.

\subsection{Numerical results}

\label{sect: numerical results}Recall from (\ref{scaling}) that the value
function satisfies%
\begin{eqnarray}
V(s,z) &=&\frac{s^{2}}{\eta \sigma ^{2}}u_{\infty }(\eta \sigma ^{2}\frac{z}{%
s})=sz\frac{u_{\infty }(\sigma ^{2}x)}{\sigma ^{2}x},  \notag \\
x &=&\eta \frac{z}{s}.  \label{eq: x2}
\end{eqnarray}%
Here $u_{\infty }$ is the solution of $\mathrm{BVP}_{[0,\infty )}$ which in
practice will be approximated by solution $\underline{\mathrm{BVP}}%
_{[0,L]}^{t}$ for sufficiently high $t$ and $L$ as described in Section \ref%
{sect: finite difference}. 
\citet{breen.al.02}
estimate linear impact of the sale of 1000 shares in a 5-minute window at
around $0.18\%$ of unaffected price. If we let $z=1$ represent 1000 shares, $%
T=1$ one year with $n=250\times 8\times 60$ trading minutes and set the
initial stock price to $s=100$ the implied value of $\eta $ turns out to be%
\begin{equation*}
\eta =0.0018\times s\times \frac{5}{n}\approx 7.5\times 10^{-6}.
\end{equation*}%
The slightly higher estimated figure of $0.3\%$ price impact from 
\citet[Figure IV]{hasbrouck.91}
results in $\eta \approx 1.25\times 10^{-5}$. We set $\sigma =0.2$ in all
examples.

Variable $x$ in equation (\ref{eq: x2}) measures percentage drop in
execution price assuming complete liquidation over one calendar year at a
constant speed (and no accruing interest). Since $sz$ is the revenue from
selling the entire inventory $z$ at price $s$ immediately and without any
price impact, $I(s,z):=1-u_{\infty }(\sigma ^{2}x)/(\sigma ^{2}x)\ $measures
the percentage drop of average per-share realized price $V(s,z)/z$ relative
to pre-trade price $s$. The quantity $I(s,z)$ is colloquially known as the
`price impact'.

From (\ref{eq: vstar}) the agent's optimal selling strategy in the original
coordinates is given by%
\begin{equation*}
v(s,z)=\frac{s-V_{z}(s,z)}{2\eta }=s\frac{1-u_{\infty }^{\prime }(\eta
\sigma ^{2}\frac{z}{s})}{2\eta }.
\end{equation*}%
The time\ to liquidation, assuming constant liquidation speed (and no
accruing interest), equals 
\begin{equation*}
\tau (s,z):=\frac{z}{v(s,z)}=\frac{2x}{1-u_{\infty }^{\prime }(\sigma ^{2}x)}%
.
\end{equation*}%
However, the actual liquidation speed is far from constant -- the asymptotic
expansion (\ref{eq: u(x) series}) shows it to be proportional to $\sqrt{z}.$
Therefore, as a rule of thumb, $\tau (s,z)$ is roughly half of the actual
average time to liquidation. This can be seen in Figure \ref{fig: 3}.

\begin{table}[t]
\centering%
\begin{tabular}{ccccccccc}
\hline\hline
& $\sigma $ & $\eta $ & $s,z$ & $\lambda $ & $r$ & $\rho $ & $a$ & $b$ \\ 
\hline
\multicolumn{1}{l}{Parametrization 1} & 0.2 & $7.5\times 10^{-6}$ & 100 & 0
& 0 & 0.05 & 2 & 0.5 \\ 
\multicolumn{1}{l}{Parametrization 2} & 0.2 & $7.5\times 10^{-6}$ & 100 & $%
-0.1$ & 0 & 0 & $-3$ & 8 \\ 
\multicolumn{1}{l}{Parametrization 3} & 0.2 & $7.5\times 10^{-6}$ & 100 & 
0.03 & 0.01 & 0.05 & 3 & $-2.5$ \\ \hline
\end{tabular}%
\caption{Parameter values used in numerical examples.}
\label{tab:parameters}
\end{table}

Table \ref{tab:parameters} shows three combinations of parameter values used
in numerical examples. Parametrization 1 has $\lambda=r=0$, meaning that the
pressure to liquidate only stems from discounting future revenues at the
rate of $\rho=0.05$. Parametrization 2 has $r=\rho=0$ and the pressure to
liquidate in this case stems from the unaffected asset price having a
negative drift of $\lambda=-0.1$. The last parametrization has positive
values of all parameters. Note that the three parametrizations also cover
the three possible combinations of signs of $a$ and $b$ which allow for $%
a+b>0$ to be satisfied.

Part (a) of Figure \ref{fig: 2} shows the per-share price impact $I(s,z) = 1-%
\frac{u(\sigma ^{2}x)}{\sigma ^{2}x}$ for the three examples. Part (b) of
the same figure shows the time to liquidation $\tau (s,z)=\frac{2x}{%
1-u^{\prime }(\sigma ^{2}x)}$.

\begin{figure}[t]
\centering%
\begin{subfigure}[b]{0.49\textwidth}
        		\centering
                \includegraphics[height=5.4cm]{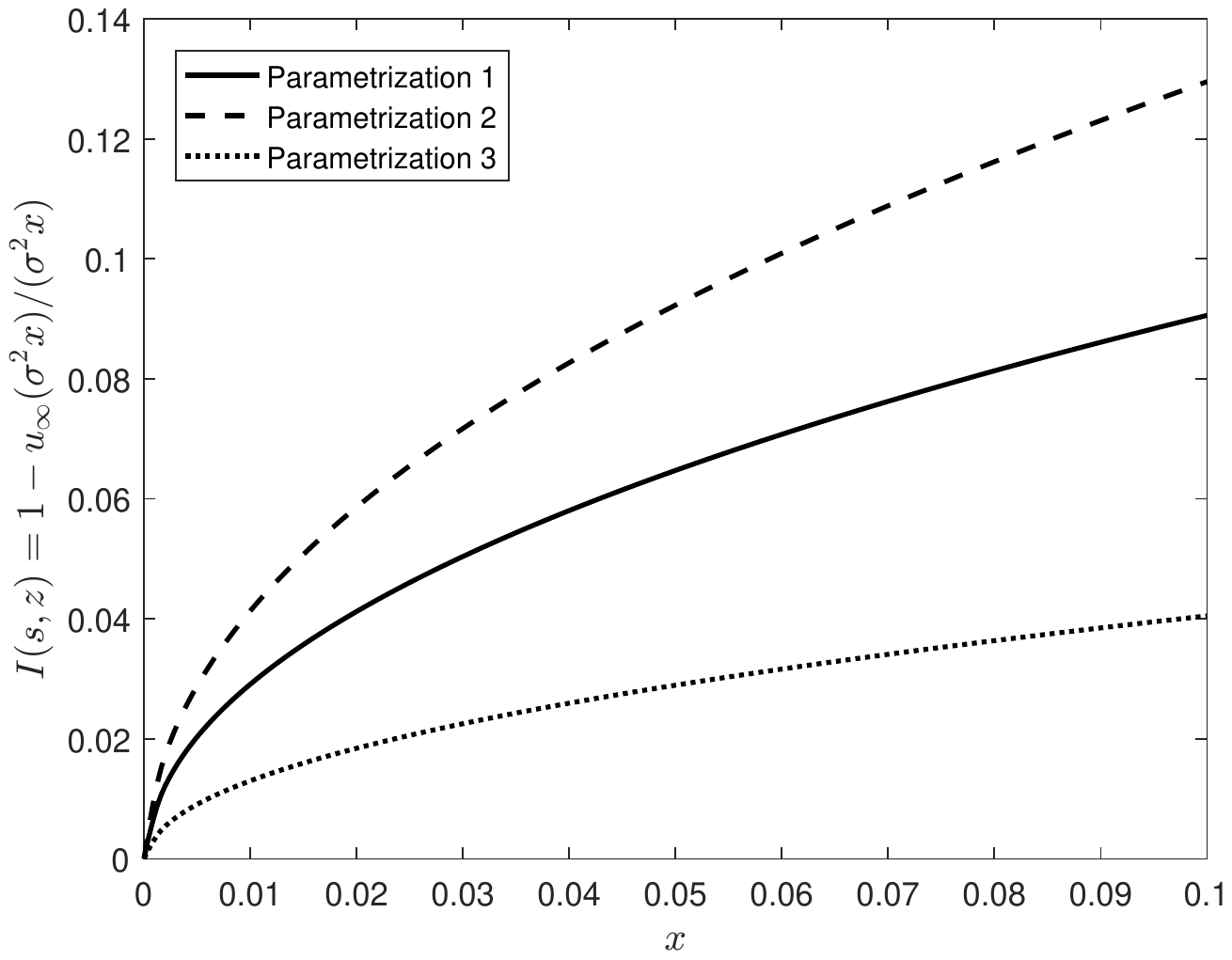}
                \caption{}
        \end{subfigure}%
\begin{subfigure}[b]{0.49\textwidth}
                \centering
                \includegraphics[height=5.4cm]{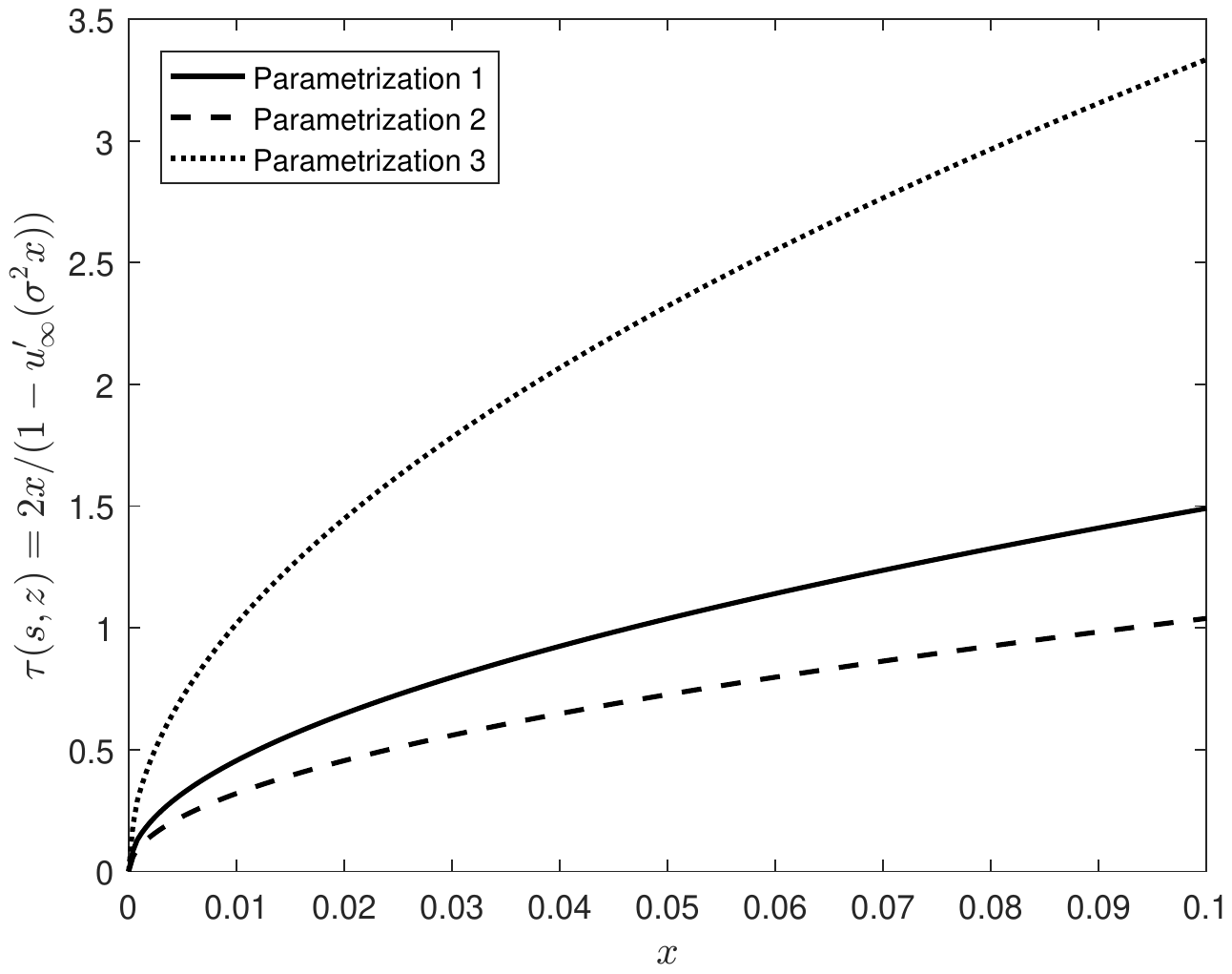}
                \caption{}
        \end{subfigure}
\caption{(a) Relative implementation shortfall; (b) Time to liquidation
assuming constant liquidation speed and no accruing interest, for three
parametrizations in Table \protect\ref{tab:parameters}.}
\label{fig: 2}
\end{figure}

Figure \ref{fig: 3} compares the time to liquidation assuming constant
liquidation speed and no accruing interest, $\tau (s,z)$, with the actual
average time to liquidation, $T(Z=0)$, which was computed based on 10,000
simulations. The initial block order size is fixed at $z=100$ corresponding
to 100,000 shares. The time to liquidation increases with stronger temporary
price impact $\eta $ and the actual actual time to liquidation is longer
than $\tau (s,z)$.

\begin{figure}[t]
\centering\includegraphics[height=5.5cm]{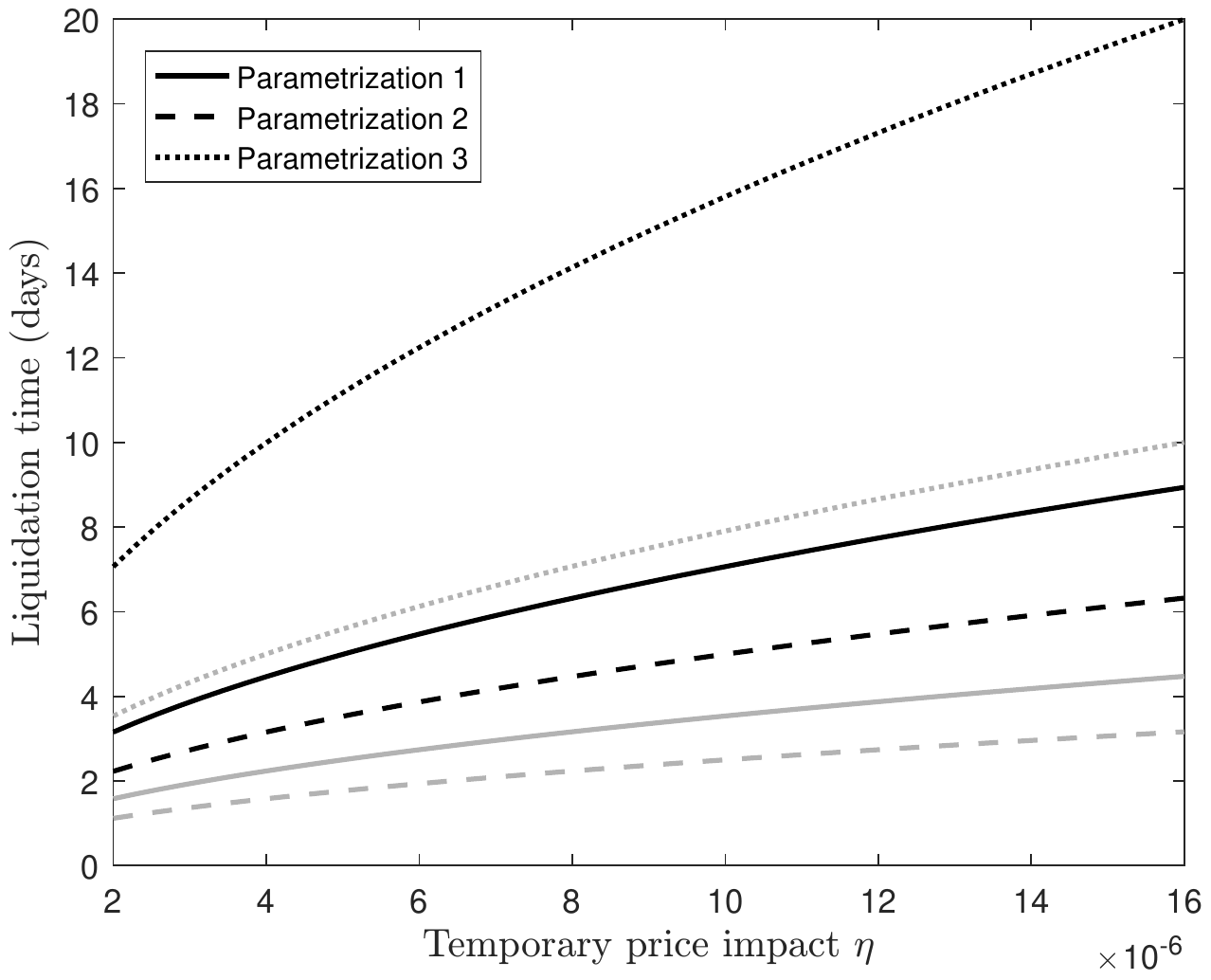}
\caption{Actual average time to liquidation, $T(Z=0)$, based on 10,000
simulations (black lines) and approximate time to liquidation, assuming
constant liquidation speed, $\protect\tau (z,s)$, (grey lines), for three
parametrizations in Table \protect\ref{tab:parameters} and changing values
of the temporary price impact parameter $\protect\eta $.}
\label{fig: 3}
\end{figure}

Figure \ref{fig: 4} shows 10,000 simulations of the liquidation with $%
s=z=100 $ and $\eta =7.5\times 10^{-6}$ calibrated from 
\citet{breen.al.02}%
. All lines are shown until the (stochastic) time of liquidation, $T(Z=0)$,
is reached. In the first column, we observe that, with each of the parameter
sets, the execution time increases when the asset price is falling. On
average, the execution takes 6.17, 4.36 and 13.80 days for the three
parametrizations in Table \ref{tab:parameters}, respectively.

\begin{figure}[t]
\centering%
\begin{subfigure}[b]{0.32\textwidth}
        		\centering
                \includegraphics[width=\textwidth]{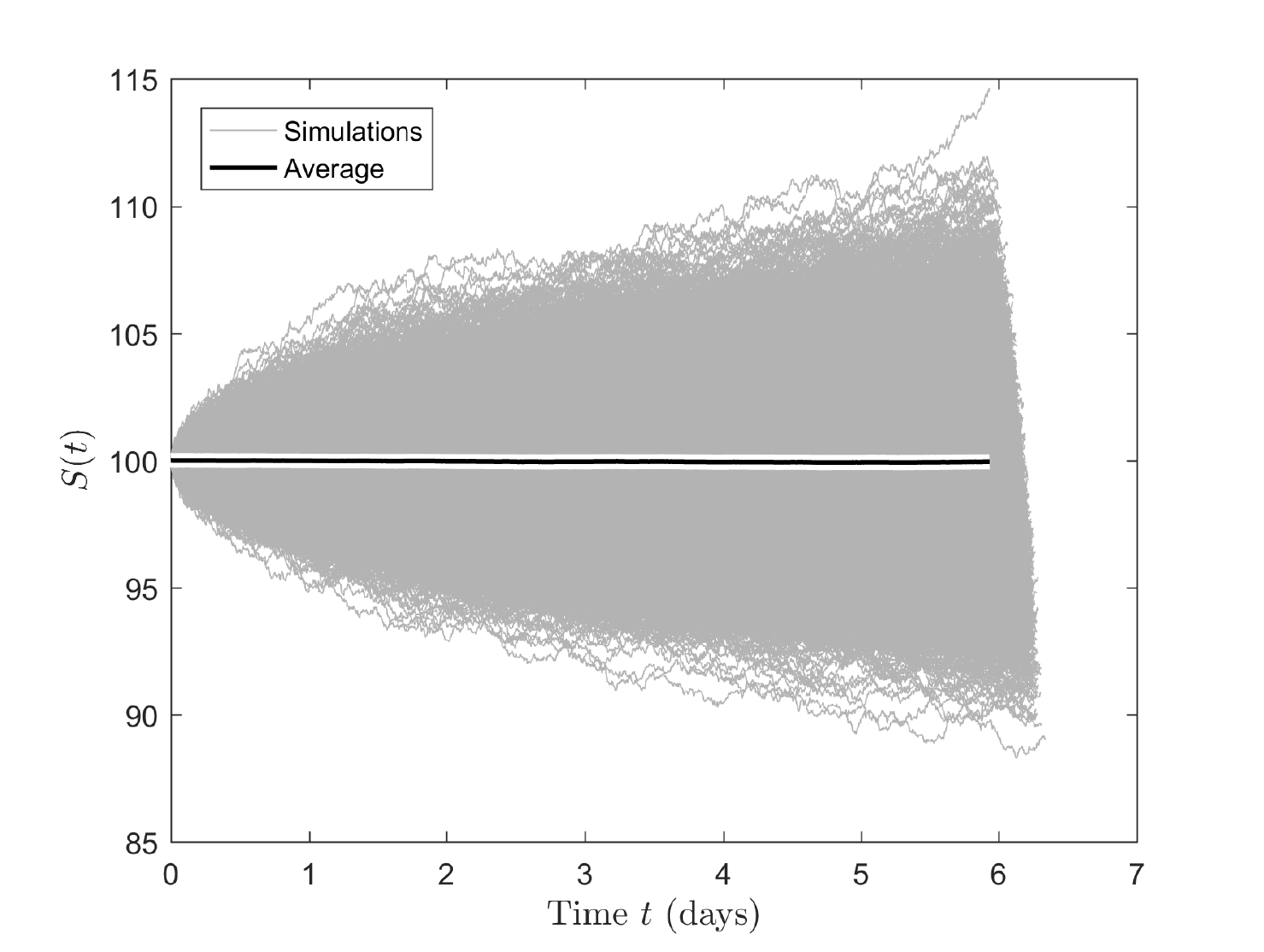}
                \caption{}
        \end{subfigure}%
\begin{subfigure}[b]{0.32\textwidth}
                \centering
                \includegraphics[width=\textwidth]{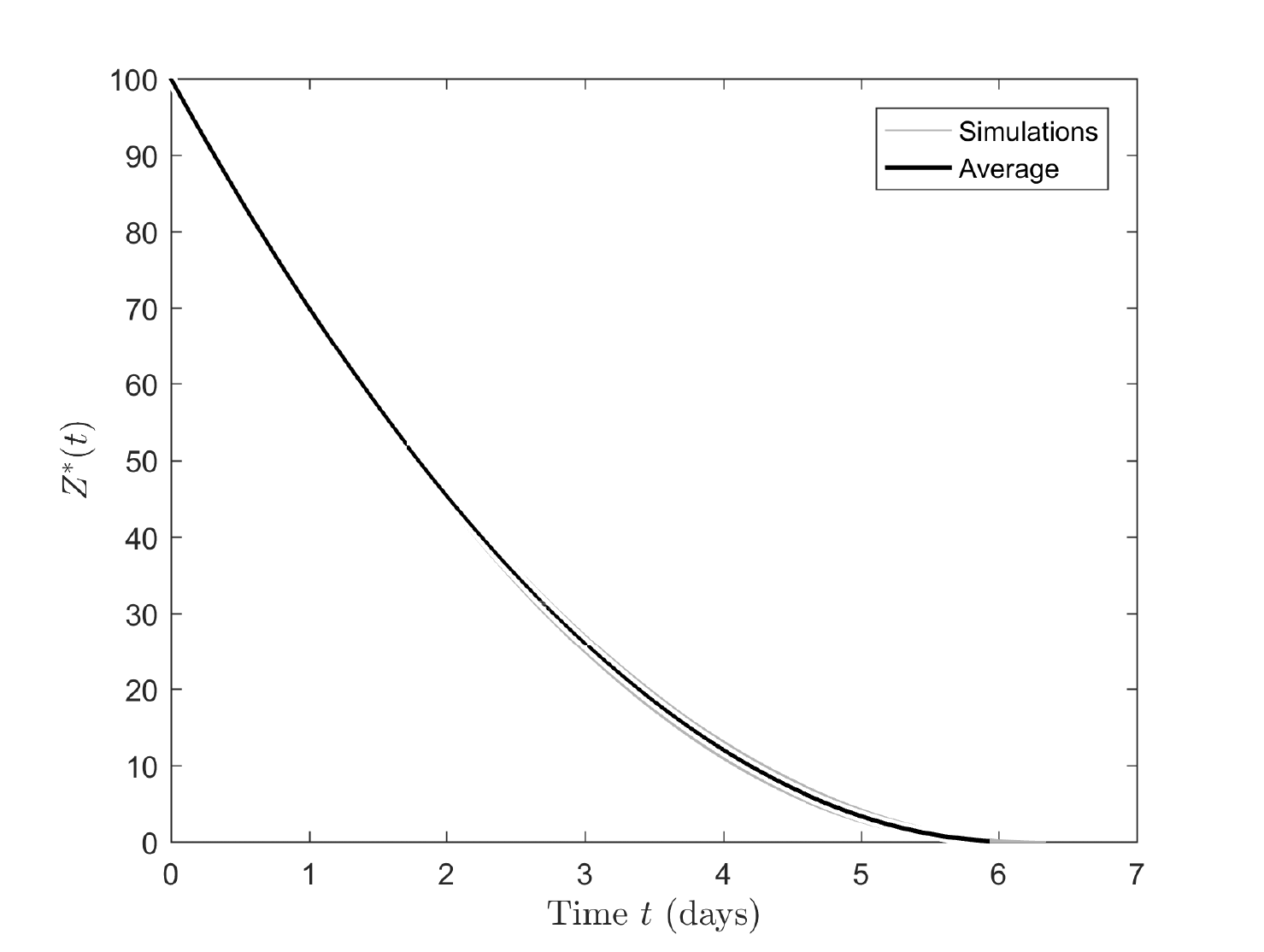}
                \caption{}
        \end{subfigure}%
\begin{subfigure}[b]{0.32\textwidth}
                \centering
                \includegraphics[width=\textwidth]{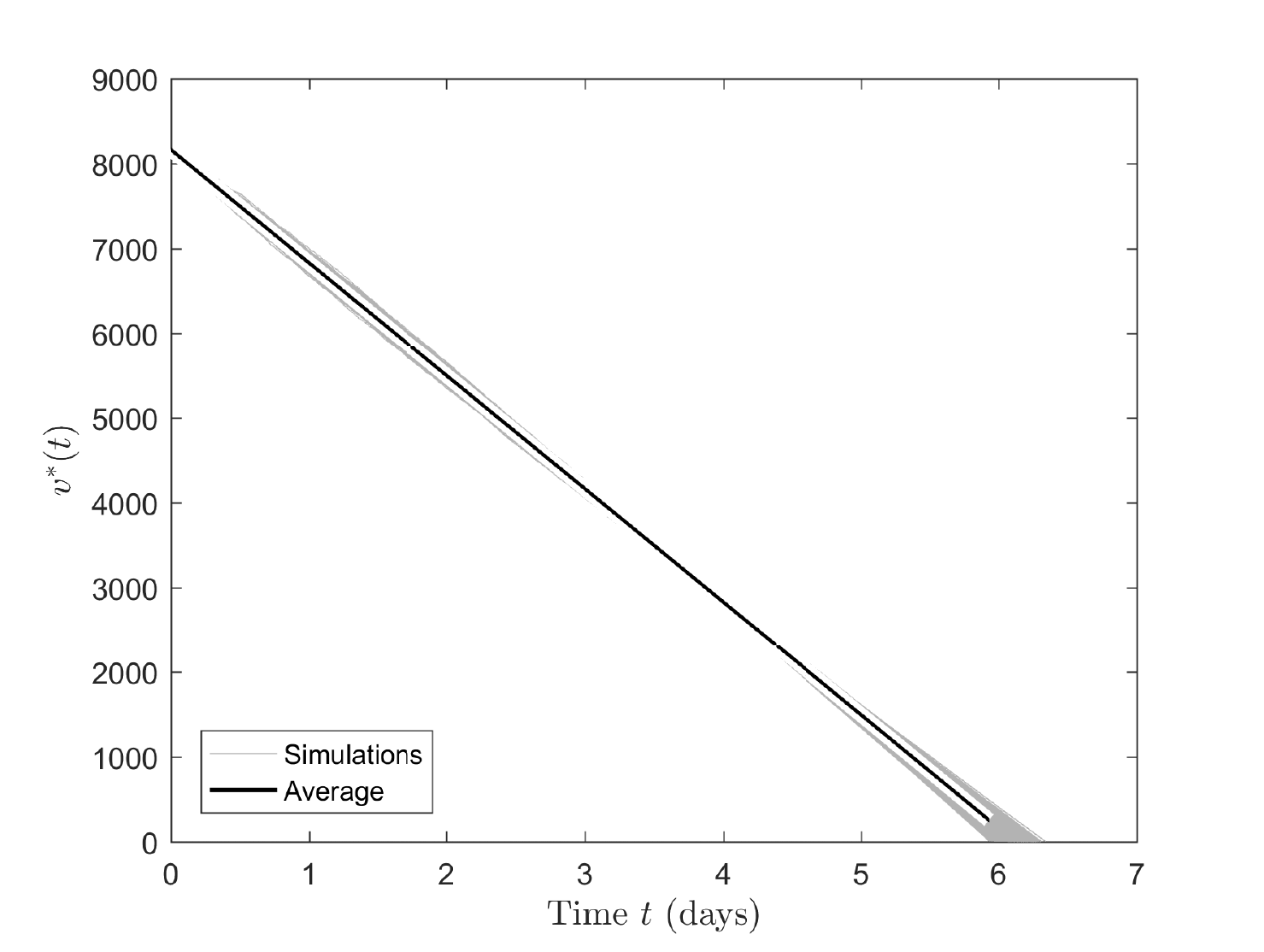}
                \caption{}
        \end{subfigure} 
\begin{subfigure}[b]{0.32\textwidth}
        		\centering
                \includegraphics[width=\textwidth]{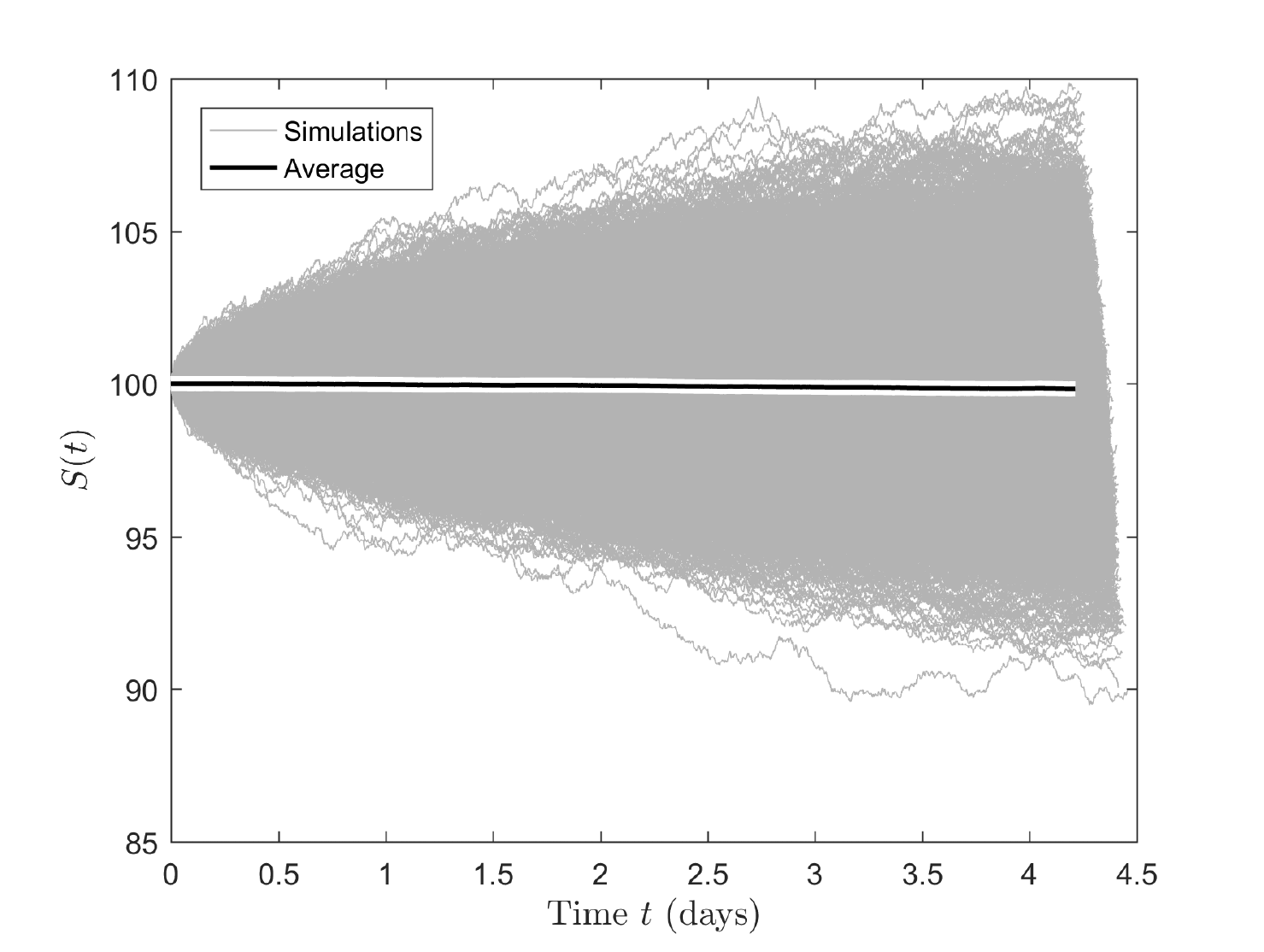}
                \caption{}
        \end{subfigure}%
\begin{subfigure}[b]{0.32\textwidth}
                \centering
                \includegraphics[width=\textwidth]{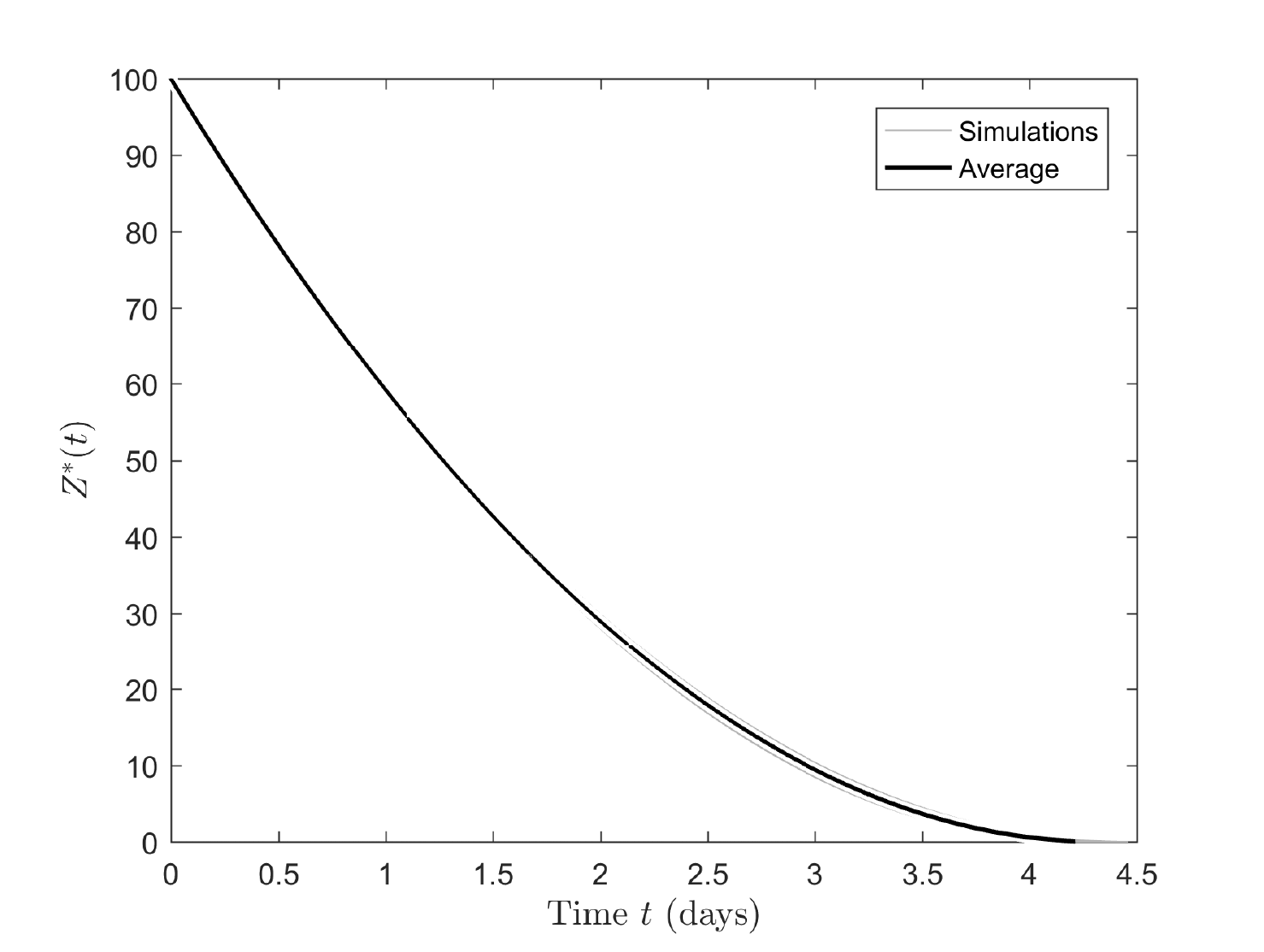}
                \caption{}
        \end{subfigure}%
\begin{subfigure}[b]{0.32\textwidth}
                \centering
                \includegraphics[width=\textwidth]{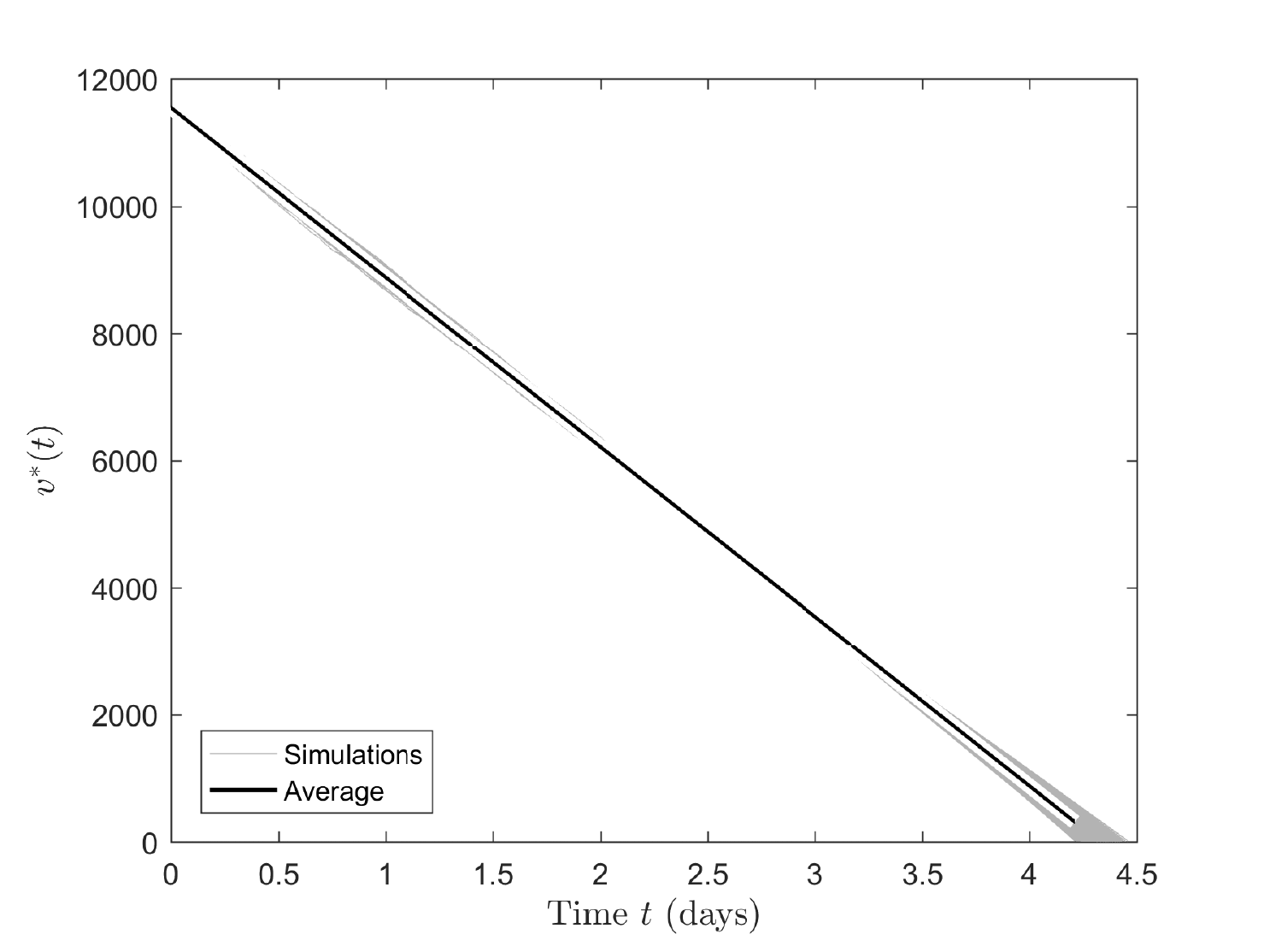}
                \caption{}
        \end{subfigure}
\begin{subfigure}[b]{0.32\textwidth}
        		\centering
                \includegraphics[width=\textwidth]{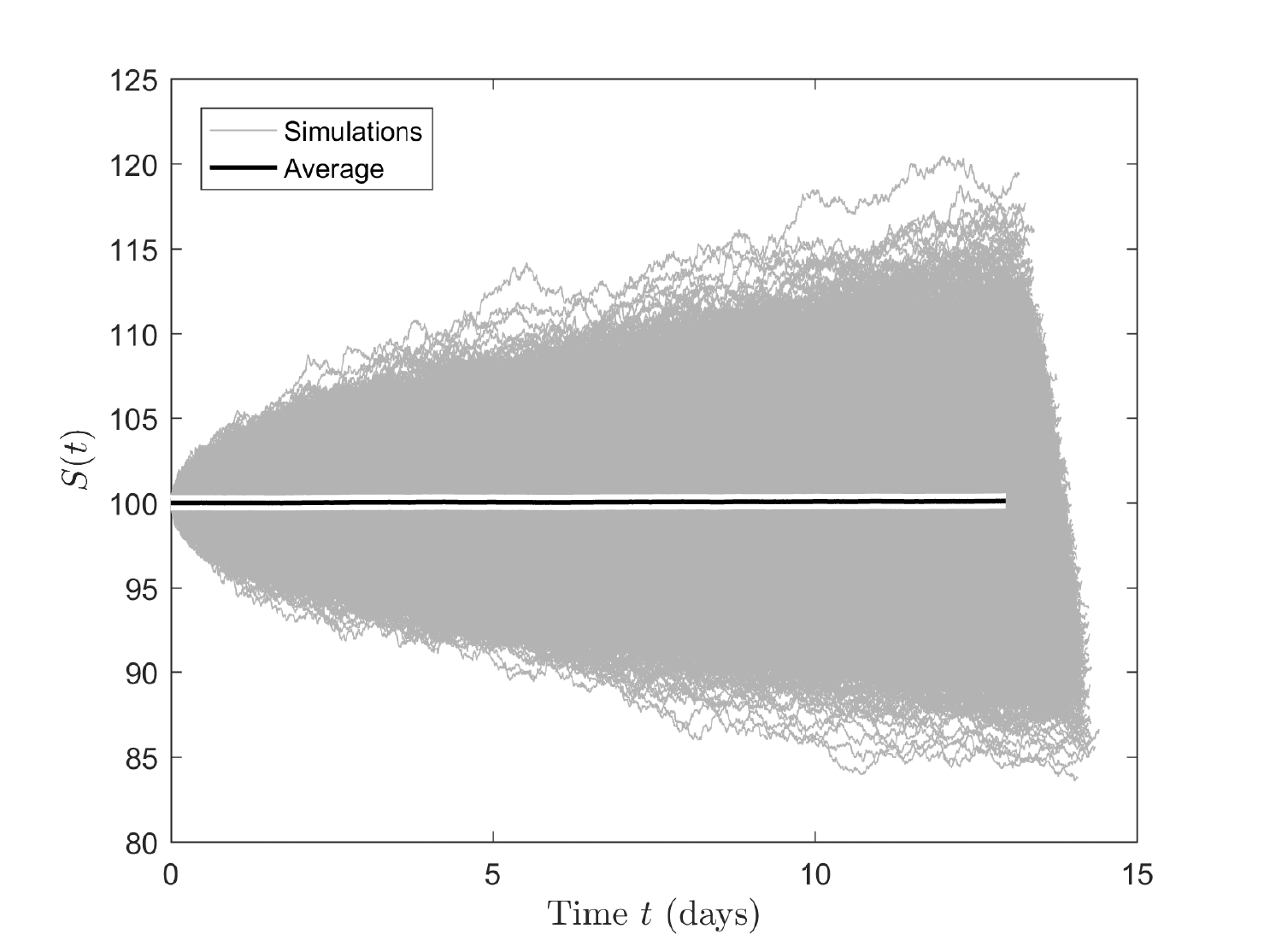}
                \caption{}
        \end{subfigure}%
\begin{subfigure}[b]{0.32\textwidth}
                \centering
                \includegraphics[width=\textwidth]{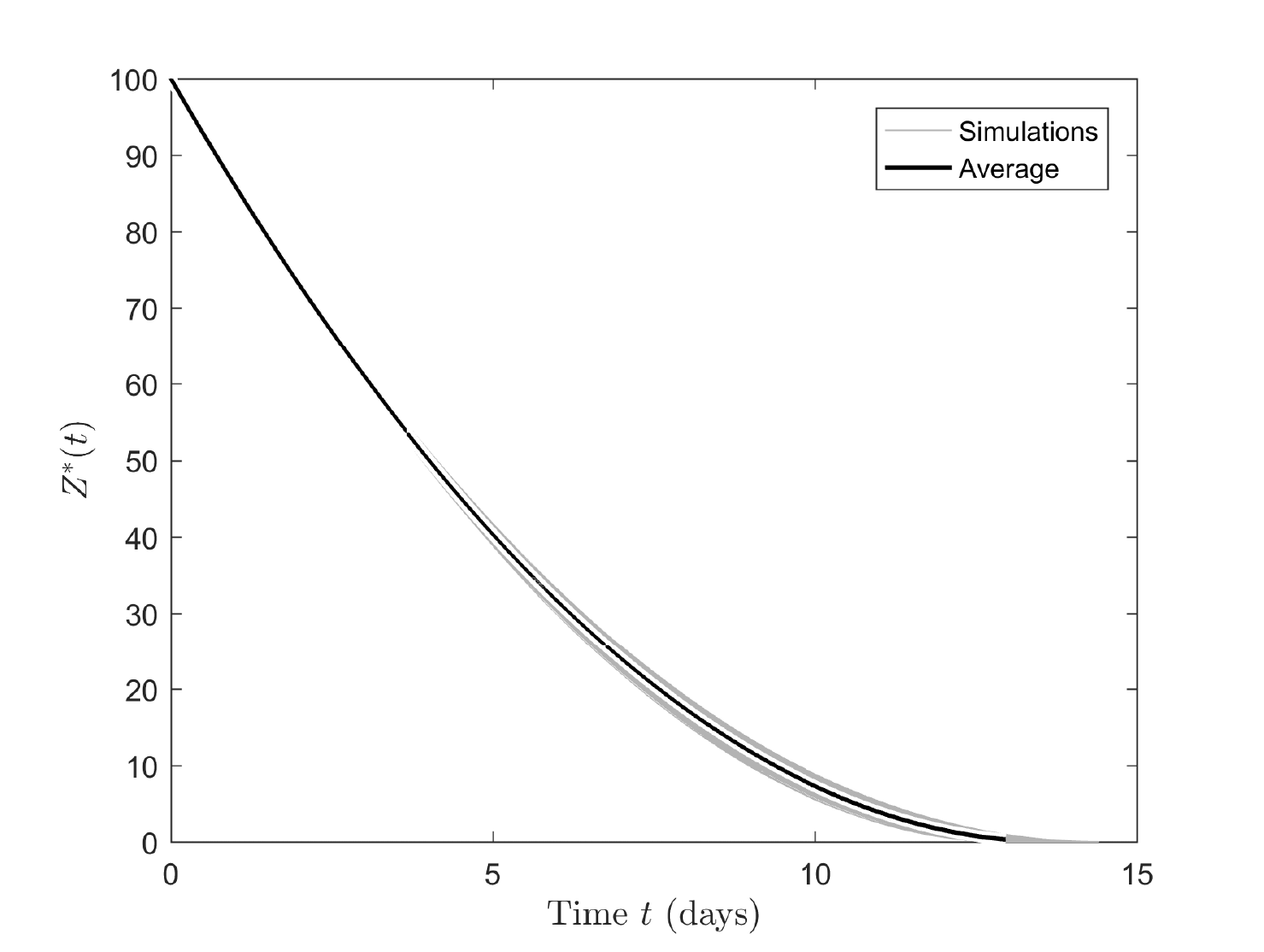}
                \caption{}
        \end{subfigure}%
\begin{subfigure}[b]{0.32\textwidth}
                \centering
                \includegraphics[width=\textwidth]{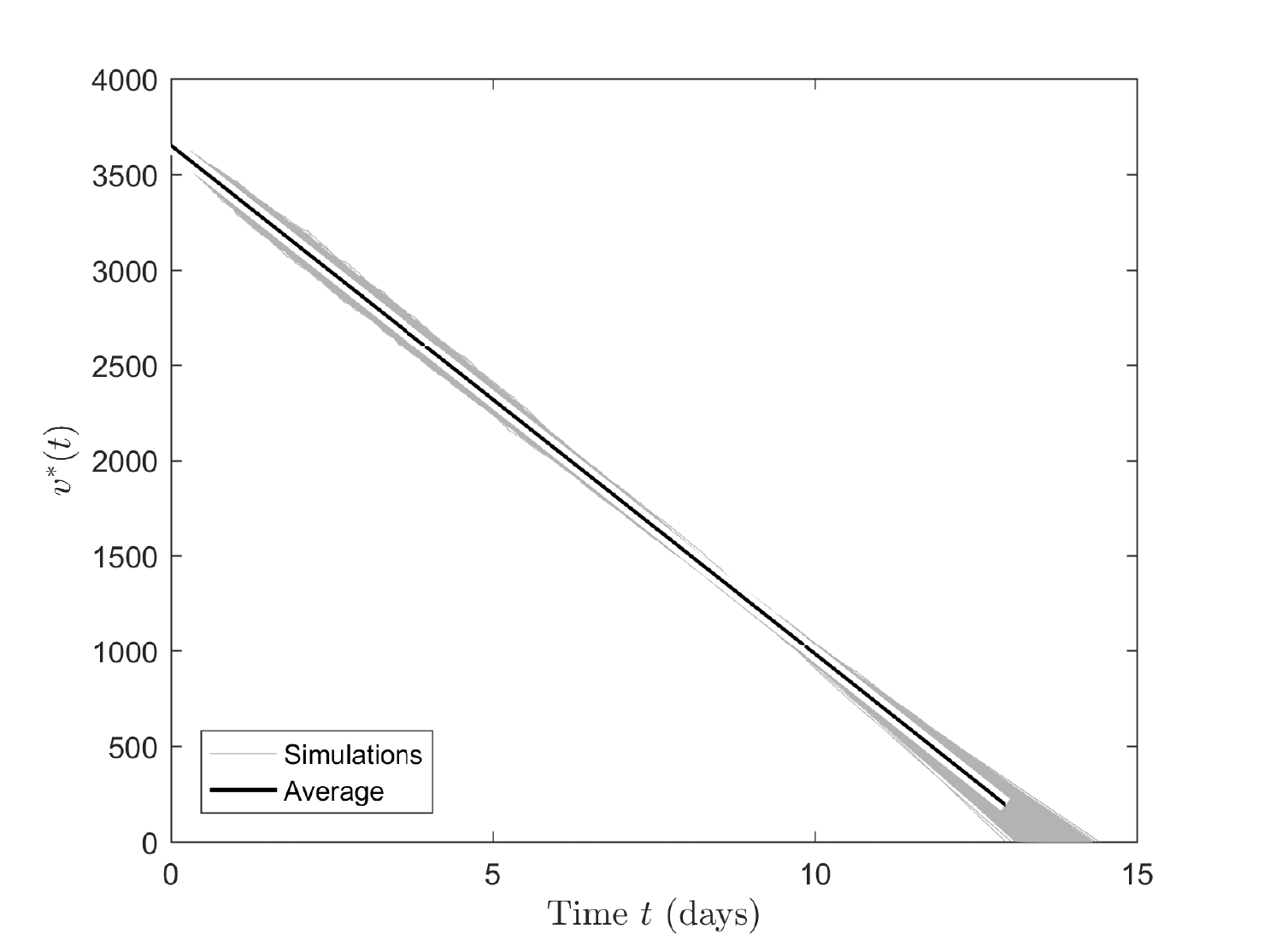}
                \caption{}
        \end{subfigure}
\caption{Each row shows 10,000 simulations of the unaffected price $S(t)$
(first column), inventory $Z^{\ast }(t)$ (second column) and the optimal
strategy $v^{\ast }(t)$ (third column), for one of the three
parametrizations in Table \protect\ref{tab:parameters}.}
\label{fig: 4}
\end{figure}

\section{Conclusions}

We have analyzed optimal liquidation of an asset whose unaffected price
drifts downwards, while assuming that short sales of the asset are ruled out
and the liquidation causes a linear temporary adverse price impact. In this
setting the liquidation time horizon becomes stochastic and is determined
endogenously as part of the optimal liquidation strategy. We have recovered
classical result from the martingale case whereby optimal liquidation always
leads to implementation shortfall, in contrast to previous studies using a
fixed time horizon. While the `raw' impact is linear the optimized impact is
asymptotically proportional to the square root of the total volume of the
order. This conclusion is well supported by empirical evidence.

The HJB\ equation of the new optimization gives rise to a boundary value
problem whose degree of singularity is not covered in the existing
literature. We have proposed a numerical scheme that overcomes the
singularity and we have provided detailed theoretical analysis of the mixed
boundary singular PDE our numerical scheme is based on.

For simplicity our work leaves out permanent impact and considers only
linear utility. We have shown in Section 2 that in the martingale case with
linear utility function the temporary and permanent impacts do not interact.
In a drifting market there will be some degree of interaction, but for
reasons given in Section 2 we suspect it to be rather weak. The precise
nature of this interaction remains an intriguing area for future research.

\textbf{Acknowlegements:} We would like to thank M. Fila, P. Pol\'{a}\v{c}%
ik, P. Quittner and M. Winkler for advice concerning Sections \ref{sect: opt}
and \ref{sect: comp}. We are grateful to two anonymous referees for detailed
comments and to participants of the London Mathematical Finance Seminar for
their feedback. P. Brunovsk\'{y} thankfully acknowledges support of VEGA
grant Nr. 1/0319/15. Thanks also go to V\'{U}B Foundation for its support of
A. \v{C}ern\'{y}'s semester visit of Comenius University in Bratislava
during which this research was initiated.

\end{document}